\documentclass[journal, onecolumn, 12pt]{IEEEtran}
\usepackage{graphicx}
\usepackage{setspace}
\usepackage{caption}
\usepackage{subcaption}
\usepackage{bm, amsmath, amssymb, amsthm}
\usepackage{amsfonts}
\usepackage {amssymb,amsmath}
\usepackage{cite}
\usepackage{epstopdf}
\usepackage{booktabs}
\usepackage{microtype}
\usepackage{color}
\DeclareGraphicsExtensions{.pdf,.png,.jpg}

\newcommand {\Define} {\stackrel {\Delta} {=}  }
\newcommand{\mya}{\mathrel{\overset{\makebox[0pt]{{\tiny(a)}}}{=}}}
\newcommand{\myb}{\mathrel{\overset{\makebox[0pt]{{\tiny(b)}}}{=}}}
\newcommand{\myc}{\mathrel{\overset{\makebox[0pt]{{\tiny(c)}}}{=}}}

\newcommand{\mygtb}{\mathrel{\overset{\makebox[0pt]{{\tiny(b)}}}{\geq}}}

\newtheorem{theorem}{Theorem}
\newtheorem{corollary}{Corollary}

\usepackage{cite}

\begin{document}

\title{OTFS Based Random Access Preamble Transmission For High Mobility Scenarios}
\author{\IEEEauthorblockN{Alok Kumar Sinha, Saif Khan Mohammed, P. Raviteja, Yi Hong and Emanuele Viterbo}
\IEEEauthorblockA{ \thanks{\copyright~ 2020 IEEE. Personal use of this material is permitted. Permission from IEEE must be obtained for all other uses, in any current or future media, including reprinting/republishing this material for advertising or promotional purposes, creating new collective works, for resale or redistribution to servers or lists, or reuse of any copyrighted component of this work in other works.}}
\IEEEauthorblockA{ \thanks{Alok Kumar Sinha is with the Bharti School of Telecommunication Technology and Management, Indian Institute of Technology Delhi (IIT Delhi) and Saif Khan Mohammed is with the Department of Electrical Engineering, IIT Delhi, New Delhi, India. Email: saifkmohammed@gmail.com. Saif Khan Mohammed is also associated with the Bharti School of Telecommunication Technology and Management, IIT Delhi. This work is an outcome of the Research and Development work undertaken in the project under the Visvesvaraya PhD Scheme of Ministry of Electronics and Information Technology, Government of India, being implemented by Digital India Corporation (formerly Media Lab Asia). This work was also supported by, the EMR funding from the Science and Engineering
Research Board (SERB), Department of Science and Technology (DST),
Government of India, and the Prof. Kishan and Pramila Gupta Chair at IIT Delhi. P. Raviteja (raviteja.patchava@monash.edu), Yi Hong (yi.hong@monash.edu) and Emanuele Viterbo (emanuele.viterbo@monash.edu) are with the Department of Electrical and Computer Systems Engineering, Monash University, Australia.}}
}
\maketitle

\vspace{-2mm}
\begin{abstract}
We consider the problem of uplink timing synchronization for Orthogonal Time Frequency Space (OTFS) modulation based systems where information is embedded in the delay-Doppler (DD) domain. For this, we propose a novel Random Access (RA) preamble waveform based on OTFS modulation. We also propose a method to estimate the round-trip propagation delay between a user terminal (UT) and the base station (BS) based on the received RA preambles in the DD domain. This estimate (known as the \emph{timing advance} estimate) is fed back to the respective UTs so that they can {\em advance} their uplink timing in order that the signal from all UTs in a cell is received at the BS in a time-synchronized manner. Through analysis and simulations we study the impact of OTFS modulation parameters of the RA preamble on the probability of timing error, which gives valuable insights on how to choose these parameters. Exhaustive numerical simulations of high mobility scenarios suggests that the timing error probability (TEP) performance of the proposed OTFS based RA is much more robust to channel induced multi-path Doppler shift when compared to the RA method in Fourth Generation (4G) systems.            
\end{abstract}

\section{Introduction}
Recently, a new modulation scheme called Orthogonal Time Frequency Space (OTFS) modulation has been introduced, where the information symbols are transmitted in the delay-Doppler (DD) domain instead of the time-frequency domain (as in Orthogonal Frequency Division Multiplexing (OFDM) systems) \cite{OTFSpaper, OTFSOFDM,Hadaniwhitepaper, SRakib, otfsmmwave}. In OTFS modulation, the delay-Doppler domain is $T$ seconds wide along the delay domain and $\Delta f = 1/T$ Hz wide along the Doppler domain. The delay domain is partitioned into $M$ sub-divisions of $T/M$ seconds each and the Doppler domain is partitioned into $N$ sub-divisions of $\Delta f/N$ Hz each. The combination of a sub-division along the delay domain and a sub-division along the Doppler domain is referred to as a delay-Doppler resource element (DDRE). It has been shown that for reliable communication of information, OTFS modulation is more robust to channel induced Doppler spread when compared to OFDM based 4G systems \cite{Hadaniwhitepaper}. This is because, unlike OFDM systems where channel induced Doppler shift spreads information sent on one sub-carrier to {\em all} sub-carriers, in OTFS modulation an information symbol transmitted on a DDRE {\em does not} spread as much interference to all other DDREs. Rather,
an information symbol sent on a particular DDRE is received mostly only on those DDREs which are separated from the transmit DDRE in the delay domain by a duration approximately equal to the delay of some channel path and in the Doppler domain by a shift approximately equal to the Doppler shift of the same channel path. Depending on the number of distinct paths, the information symbol could be received on multiple DDREs, which can then be coherently combined to achieve delay-Doppler diversity \cite{SRakib}. Further, each DD domain information symbol sees the same constant channel gain, which greatly simplifies the transmitter and receiver design \cite{chocks,emanuele, emanuele2, Farhang, EmanueleTWC, mimootfs}. The constant channel gain across the entire DD domain also helps in reducing the  overhead of frequent channel estimation and feedback. Pilot aided channel estimation in the delay-Doppler domain has been considered in \cite{EviterboEst, ChocksEst}. Further, OTFS modulation can be implemented as a precoding operation (from DD domain to time-frequency domain) followed by OFDM modulation (from time-frequency domain to time-domain), and similarly the received time-frequency OFDM signal can be converted back to the DD domain \cite{Hadaniwhitepaper}.\footnote{\footnotesize{This compatibility of OTFS modulation with existing OFDM based 4G/5G systems will be useful for industry, as it enables the use of OTFS modulation for high mobility scenarios like high speed train, vehicle-to-vehicle communication, etc., \cite{OTFSOFDM,Hadaniwhitepaper}.}}    

Several multiple-access schemes have also been proposed for OTFS based systems in \cite{patent2, OTFSMA, chocksMA }. In \cite{patent2}, the DDREs allocated to different user terminals (UTs) is separated by guard bands in order to reduce multi-user interference (MUI). However, with this approach the required size of the guard bands along the delay domain would be large, especially when the cell size is large (as in rural areas). This is because, in the absence of uplink timing synchronization, the difference between the time of arrival of uplink signals from two different UTs can be as high as the round-trip propagation delay between the base station (BS) and a cell-edge UT. With uplink timing synchronization, the UTs adjust their uplink timing such that their signals arrive at the BS in a time-synchronized manner, i.e., the difference between the time of arrival of signals from any two UTs at the BS is no more than the delay spread of the channel. Since guard bands are an overhead, uplink timing synchronization is necessary for guard-band based multiple-access methods.  

In Fourth Generation (4G) systems, uplink timing synchronization of a UT is achieved through estimation of the round-trip propagation delay (also known as \emph{timing advance}) between that UT and the BS. This estimation is based on the random access (RA) signal received at the BS and is followed by uplink timing correction at each UT based on the
timing advance (TA) estimate fedback to that UT by the BS \cite{4GLTE}. In 4G-LTE (Long Term Evolution) systems, the presence of carrier frequency offset (e.g., due to Doppler spread in high mobility scenarios) is known to significantly deteriorate the accuracy of the TA estimate acquired at the BS \cite{cfora},\cite{LTE4GTA}, due to which the required cyclic prefix length of OFDM symbols needs to be significantly higher than the channel delay spread so as to avoid inter-symbol interference.
This increases the cyclic prefix overhead and therefore reduces the overall system throughput. Fifth Generation (5G) systems are expected to support high throughput rate at even higher mobility when compared to 4G systems \cite{IMT2020}, and therefore there is a need to reconsider the design of RA preamble waveform such that TA estimation is not sensitive to Doppler spread.

\begin{figure}[h]
	\vspace{-0.2 cm}
	\centering
	\includegraphics[width= 4.0 in, height= 3.0 in]{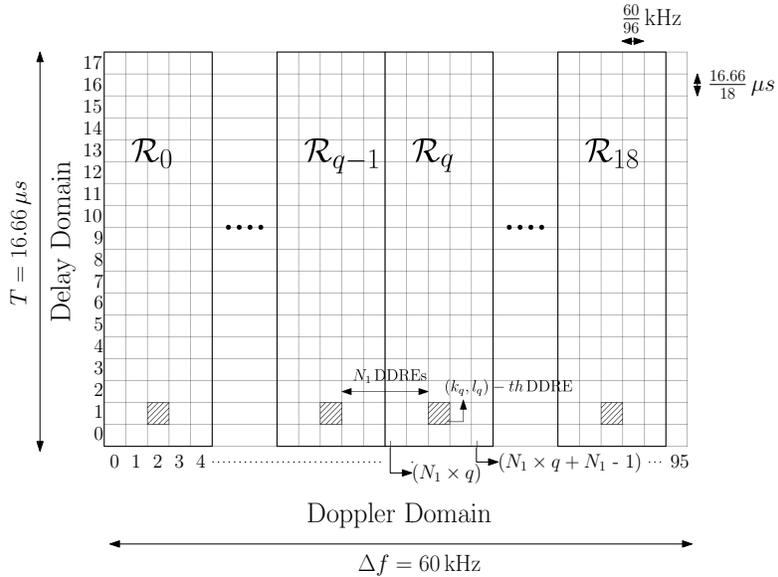}
	\vspace{-0.15 cm}
	\caption{Illustration of the proposed RA preambles in the DD domain. $T = 16.66 \mu s$, $\Delta f = 1/T = 60$ KHz, $M = 18$, $N = 96, N_1 = 5$. ${\mathcal R}_q, q = 0,1,\cdots, 18$ are the DDRE groups corresponding to the $19$ RA preambles. ${\mathcal R}_q$ and ${\mathcal R}_{q-1}$ are adjacent DDRE groups.} 
	\vspace{-0.2cm}
	\label{figrarq}
\end{figure}
{
\begin{table}[t]
\centering
\caption{Notations}
{\small
\begin{tabular}{||c|c||}
\hline
Variable & Description     \nonumber \\
\hline
$T$, $\Delta f$ &  {Delay and Doppler domain width respectively} \nonumber \\
\hline
$M$, $N$ &  {No. of sub-divisions of the delay and} \nonumber \\
&  Doppler domain respectively  \nonumber \\
\hline
$T_c$ &  {Time allocated for RA preamble transmission} \nonumber \\
& (excluding prefix and suffix) \nonumber \\
\hline
$B_c$ &  Bandwidth allocated for RA preamble transmission \nonumber \\
\hline
$G$ &  Maximum possible round-trip propagation delay \nonumber \\
&  between the BS and any UT in the cell \nonumber \\
\hline
$\nu_{max}$ &  Maximum Doppler shift of any channel path \nonumber \\
& between the BS and any UT in the cell \nonumber \\
\hline
$Q$ &  Number of UTs \nonumber \\
\hline
$R$ &  Number of RA preambles \nonumber \\
\hline
${\mathcal R}_q$ &  DDRE group allocated for  \nonumber \\
&  transmission of the $q$-th RA preamble \nonumber \\
\hline
$N_1$ &  Number of DDREs along the Doppler domain  \nonumber \\
&  in each DDRE group \nonumber \\
\hline
$P_e$ & Probability of missed detection of RA preamble \nonumber \\
& (also referred to as Timing Error Probability) \nonumber \\
\hline
$P_{fa}$ &  Probability of False Alarm \nonumber \\
\hline
$E$ &  Energy of each RA preamble \nonumber \\
\hline
$\mu$ &  False alarm threshold \nonumber \\
\hline
$\widehat { \mathrm { TA } } _ { q }$ &  Proposed timing advance (TA) estimate \nonumber \\
& for the $q$-th RA preamble \nonumber \\
\hline
$N_o$ & Power spectral density of AWGN at BS receiver \nonumber \\
\hline
$\rho$ &  Ratio of the transmitted RA preamble signal \nonumber \\
& power to the noise power at the BS receiver \nonumber \\
\hline
$k_q,l_q$ &  The Doppler domain and delay domain indices \nonumber \\
& of the DDRE where the energy of the $q$-th RA \nonumber \\
&  preamble is localized  \nonumber \\
\hline
\end{tabular}
\vspace{-2mm}
\normalsize}
\label{tabnotations}
\end{table}}
In this paper, we propose OTFS modulation based RA preamble waveform which is shown to achieve more accurate uplink timing synchronization when compared to 4G-LTE systems. For notational clarity, in Table-\ref{tabnotations} we list the meaning of variables which have been frequently used in this paper to describe the proposed TA estimation and RA preamble detection method.
Our specific contributions are:
\begin{enumerate}
\item In Section \ref{sec3} we propose OTFS modulation based RA preamble waveforms, where the RA preambles are allocated non-overlapping contiguous rectangular DDRE groups in the DD domain (a DDRE group is a collection of DDREs). For each RA preamble, energy is transmitted on only one DDRE in its group. Each DDRE group spans the entire delay domain, but is restricted to a subset of $N_1$ DDREs along the Doppler domain (see Fig.~\ref{figrarq}).
The number of RA preambles is $R = \lfloor N/N_1 \rfloor$ where $N$ is the total number of DDREs along the Doppler domain. We also propose a TA estimation method to estimate the round-trip propagation delay between a UT and the BS.
\item We derive the expression for the received RA signal in the delay-Doppler domain and subsequently we characterize the average probability of missed detection of an RA preamble transmitted by a UT (referred to as timing error probability (TEP)).
The definition of successful detection of an RA preamble transmitted by a UT is motivated by practical considerations of the OTFS waveform. Successful detection of an RA preamble transmitted by a UT is said to happen if the received RA preamble energy is greater than or equal to a threshold and the proposed TA estimate lies between the smallest and the greatest path-delay between the BS and that UT. The TEP expression is a  sum of several terms consisting of a term for the single-UT scenario and other terms for the multi-UT scenario.   
\item In Section \ref{secsingleUT}, we study the TEP for the single-UT scenario from where it is observed that the proposed TA estimation is almost insensitive to the multi-path Doppler shift. This explains the robustness of the proposed OTFS based RA method. 
\item In Section \ref{secmultiUT} we show that in a multi-UT scenario, TEP is limited by, i) the interference between RA preambles transmitted on two adjacent DDRE groups and, ii) collisions (i.e., if two or more UTs transmit the same RA preamble). We derive an analytical expression for TEP in the presence of collisions, which is in fact a lower bound to the overall TEP.
\item In Section \ref{secmultiUT} we also study the near-far scenario where a UT near the BS and another UT at the cell-edge transmit RA preambles on adjacent DDRE groups. In such a scenario, the RA preamble energy received from a UT near the BS leaks into an adjacent DDRE group of a RA preamble received from the UT at the cell-edge. This can result in a timing error event for the UT at the cell-edge. We show that this leakage can be reduced by appropriately windowing the received RA preamble in the time-domain and choosing a sufficiently large DDRE group width $N_1$.        
\item In Section \ref{sectionSelParam}, for a desired TEP, we propose the selection of the design parameters $(T, \Delta f, M, N, N_1)$ for the OTFS based RA preamble.   
\item  Exhaustive numerical simulations in Section \ref{simsection} reveal that in high mobility scenarios, the proposed OTFS based RA method achieves significantly smaller TEP when compared to the RA method in 4G-LTE. Further, with increasing maximum Doppler shift $\nu_{max}$, TEP of the RA method in 4G-LTE degrades severely when compared to that of the proposed OTFS based RA method.    
\end{enumerate}  
\section{SYSTEM MODEL AND OTFS MODULATION}
\label{sysmodel_sec}
In this paper we consider a single-cell multi-user scenario where multiple single-antenna user terminals (UTs) can initiate random access to a single-antenna base station (BS) on a common dedicated physical resource. We consider a single antenna at the BS and at the UTs, as the main objective of this paper is to propose a novel OTFS based RA preamble waveform and TA estimation.
Having multiple antennas at the BS and at the UTs is likely to further improve the performance of the proposed OTFS based RA preamble waveform and TA estimation, but this will not change
our main conclusion that the proposed OTFS based RA preamble waveform and TA estimation is robust to channel induced Doppler shift in high mobility scenarios.

Let the number of UTs requesting simultaneous random access be denoted by $Q$. The random access (RA) preamble transmitted by the $q$-th UT is denoted by $s_q(t), q=0,1,\cdots, (Q-1)$. The RA preamble is chosen randomly by the UT from a set of permissible RA preambles. With mobile UTs, the channel from each UT to the BS is time-varying. We therefore use a delay-Doppler (DD) representation of the channel from each UT to the BS. Let the DD channel for the $q$-th UT be given by \cite{chocks,  emanuele, EmanueleTWC}
\begin{eqnarray} 
\label{equation1}
h_{q}(\tau, \nu)=\sum_{i=1}^{L_{q}} h_{q, i} \, \delta(\tau- \tau_{q, i}) \, \delta\left(\nu-\nu_{q, i}\right)
\end{eqnarray}
where $\delta(.)$ is the impulse function, $h_{q,i},\tau_{q,i}$ and $\nu_{q,i}$ are the complex channel gain, the delay and Doppler along the $i$-th multipath and $L_q$ is the number of paths.
Also
{\vspace{-2mm}
\begin{eqnarray} 
\label{equation2}
\mathbb{E}\Big[\sum_{i=1}^{L_{q}}\left|h_{q, i}\right|^{2}\Big]= \beta_q,\,\,\, q=0,1, \cdots(Q-1)
\end{eqnarray}}
where $\beta_q$ models the path-loss between the $q$-th UT and the BS.
Further
{
\vspace{-1mm}
\begin{eqnarray}
\label{equation3} 
0 \leq \tau_{q, i} \leq \tau_{\max },\,\,\,\left|\nu_{q, i}\right| \leq \nu_{\max }
\end{eqnarray}
}
where $\tau_{max}$ is the maximum possible round-trip propagation delay between any UT in the cell and the BS, and $\nu_{max}$ is the maximum possible Doppler shift of any channel path between any UT in the cell and the BS.
Also, without loss of generality, let
{
\vspace{-2mm}
\begin{eqnarray}
\label{equation4}
0 \leq \tau_{q, 1} \leq \tau_{q, 2} \leq \ldots \leq \tau_{q, L_{q}}.
\end{eqnarray}
}
The time-domain signal received at the BS is given by

{\small
\vspace{-4mm}
\begin{eqnarray}
\label{equation5}
r ( t )  & =  & \sum _ { q = 0 } ^ { Q - 1 } \iint h _ { q } ( \tau , \nu ) s _ { q } ( t - \tau ) e ^ { j 2 \pi \nu ( t - \tau ) } d \nu d \tau + w ( t )  \nonumber \\
& \mya &    \sum _ { q = 0} ^ { Q-1 } \sum _ { i = 1 } ^ { L _ { q } } h _ { q , i }  \, s _ { q } \left( t - \tau _ { q , i } \right) e ^ { j 2 \pi \nu _ { q , i } \left( t - \tau _ { q , i } \right) } + w ( t )
\end{eqnarray}
\normalsize}
where step (a) follows from (\ref{equation1}) and $w(t)$ is the AWGN at the BS modeled as a white Gaussian random process having zero mean and power spectral density $N_o$. 

OTFS modulation has been recently proposed as an alternative to OFDM for data communication, which has been shown to exhibit better robustness towards Doppler shifts in high mobility scenarios when compared to OFDM \cite{OTFSOFDM, Hadaniwhitepaper}.
In OTFS modulation, the information symbols are sent in the delay-Doppler (DD) domain which is $T$ sec wide along the delay domain and $\Delta f = 1/T$ Hz wide along the Doppler domain. The delay domain is divided into $M$ sub-divisions, with each sub-division being $T/M$ sec wide. The Doppler domain is divided into $N$ sub-divisions, with each sub-division being $\Delta f/N$ Hz wide. Therefore, the DD domain is divided into $M \times N$ delay-Doppler Resource Elements (DDREs), i.e., each DDRE is $T/M $ sec along the delay domain and $ \Delta f/N$ Hz along the Doppler domain. Let ${\small x_q[k,l]}$ denote an information symbol transmitted by the $q$-th UT on the $(k,l)$-th DDRE. The $(k,l)$-th DDRE comprises of the interval $[ lT/M , (l+1)T/M)$ sec along the delay domain and the interval $[k\Delta f/N , (k+1)\Delta f/N)$ Hz along the Doppler domain. Using the Inverse Symplectic Finite Fourier Transform (Inverse SFFT), the DD information symbols are firstly transformed to the Time-Frequency (TF) domain ($NT$ sec $\times$ $M\Delta f$ Hz) \cite{OTFSOFDM, Hadaniwhitepaper}.

The frequency domain is divided into $M$ sub-divisions and the time domain is divided into $N$ sub-divisions, i.e., the entire TF domain is subdivided into $MN$ sub-divisions, where each sub-division of the TF domain is $\Delta f$ Hz wide along the frequency domain and $T$ sec wide along the time domain. Each such sub-division of the TF domain is referred to as a Time Frequency Resource Element (TFRE). The $(m,n)$-th TFRE comprises of the interval $[m\Delta f , (m+1)\Delta f)$ along the frequency domain and the interval $[nT , (n+1)T)$ along the time domain. The modulated TF symbol transmitted by the $q$-th UT on the $(m,n)$-th TFRE is given by
{
\vspace{-1mm}
\begin{eqnarray}
\label{equation17}
X_q[n,m]    =    \frac{1}{MN} \sum_{k=0}^{N-1}\sum_{l=0}^{M-1}  \hspace{-1mm} x_q[k,l] \, e^{-j 2 \pi {\big (}   \frac{ml}{M}  - \frac{n k}{N} {\big )} }  , \nonumber \\
  m=0,\cdots, (M-1) , n=0,\cdots,(N-1).
\end{eqnarray}
}
These are then converted to time domain and transmitted \cite{OTFSOFDM, Hadaniwhitepaper}, i.e.
{
	\begin{eqnarray}
	\label{equation18}
	{s}_q(t) & \hspace{-2mm} = &  \hspace{-2mm} \sum_{m=0}^{M-1}\sum_{n=0}^{N-1} X_q[n,m] g_{tx}(t - nT) e^{j 2 \pi m \Delta f (t - nT)}
	\end{eqnarray}    
}
where $g_{tx}(t)$ is the rectangular transmit pulse given by
\begin{eqnarray}
\label{gtxpulse_eqn}
g_{tx}(t) \Define \left\{
  \begin{array}{@{}ll@{}}
   \frac{1}{\sqrt{T}} \, &, \, \mbox{if} \,\,\, 0 \leq t < T \\
   0  \, & \,,\, \mbox{otherwise}.
  \end{array}\right.
\end{eqnarray}
The total time duration and bandwidth of ${s}_q(t)$ is $NT$ seconds and $M \Delta f$ Hz, respectively \cite{OTFSOFDM, Hadaniwhitepaper}. 
At the BS, the received time-domain signal $ r(t) $ is transformed to the TF domain, i.e.
	
{\vspace{-4mm}
\small
\begin{eqnarray}
\label{equation20}
Y[n,m]  = W_{rx}[n,m] \int g_{rx}^{*}(t - nT) r(t) e^{-j 2 \pi m \Delta f (t - nT)} dt
\end{eqnarray}
\normalsize
}
where $g_{rx}(\cdot)$ is the receive pulse, which is taken to be the same as the transmit pulse $g_{tx}(\cdot)$ and $W_{rx}[n,m]$ is a suitable receive windowing sequence\footnote{\footnotesize{In Section \ref{sec3} we show that time-domain receiver windowing
can reduce the energy in the side lobes of the received RA preamble along the Doppler domain, at the cost of a wider main lobe.
Reduction in side-lobe energy level implies reduction in interference between adjacent RA preambles, whereas a wider main lobe implies a larger required width $N_1$
of each RA preamble (see Fig. \ref{figrarq}). We consider windowing waveforms which are commonly used in spectral analysis,
to achieve a trade-off between lower side-lobe energy level and a wider main lobe (e.g.,
Hamming window, Blackman-Harris window).}} in the TF domain \cite{OTFSOFDM}.
This received TF domain signal is then transformed back to the DD domain through SFFT, i.e.

{\vspace{-4mm}
\small
\begin{eqnarray}
\label{equation21}
{\widehat x}[k,l] & = &  \sum_{n=0}^{N-1} \sum_{m=0}^{M-1}  Y[n,m] \, e^{j 2 \pi {\Big (}   \frac{ml}{M}  - \frac{n k}{N} {\Big )}}.
\end{eqnarray}
\normalsize
}
In order to avoid multiuser interference (MUI), $(T, \Delta f = 1/T)$ must satisfy \cite{OTFSOFDM, Hadaniwhitepaper, chocks, EmanueleTWC},  
{
\begin{eqnarray}
\label{equation22}
\tau_{max} < T \,,\,  2\nu_{max} < \Delta f.
\end{eqnarray}
}
{
\begin{figure}
\centering
	\begin{subfigure}[b]{0.8 \linewidth}
		\includegraphics[width=4.5 in, height = 3.5 in]{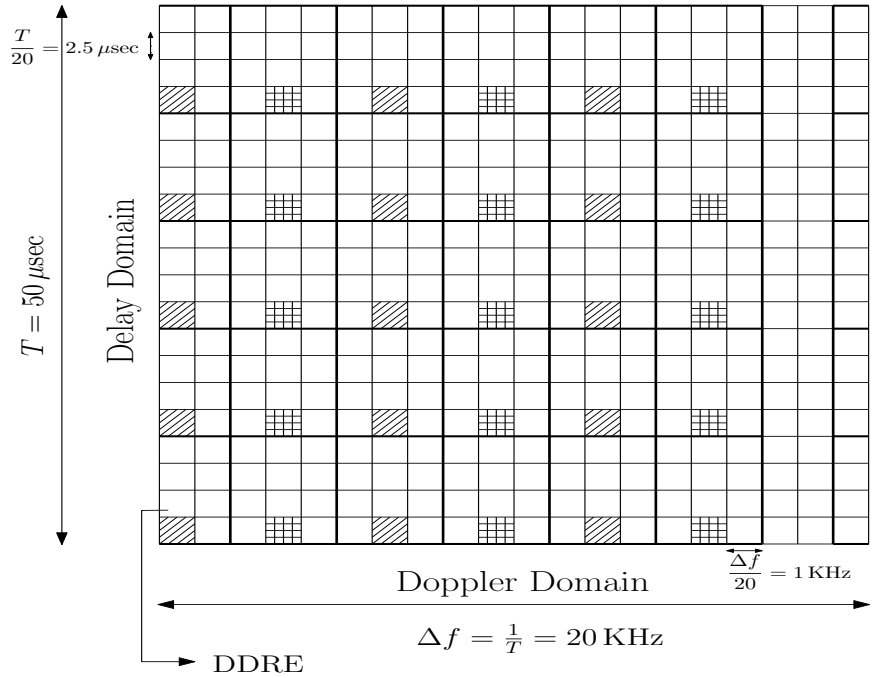}
		\caption{\centering \small{Information Tx. by two UTs in the DD domain.}}
		\label{fig1sub1}
	\end{subfigure} 
	
	\vspace{4mm}
	
	\begin{subfigure}[b]{0.8\linewidth}
		\includegraphics[width=4.5 in, height = 3.5 in]{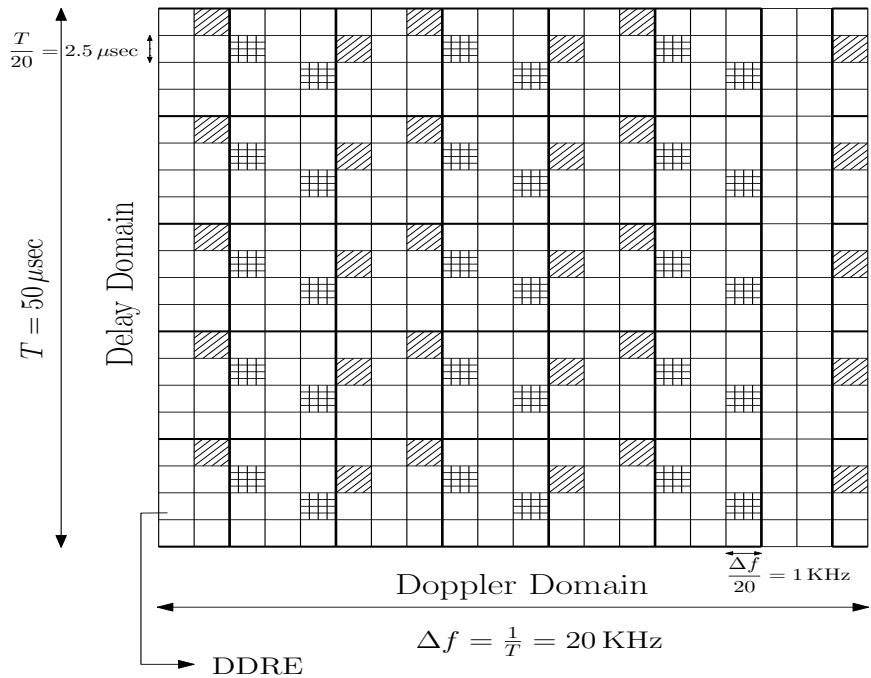}
		\caption{\centering  {\small Information Rx. at the BS in the DD domain.}}
		\label{fig1sub2}
	\end{subfigure}
	\caption{{\small Two user uplink communication in the absence of uplink timing synchronization. Information symbols are transmitted on only $30$ out of the total $M \times N = 400$ DDREs.  ($M = N = 20$, $Q = 2, T = 50 \mu s$, $\Delta f = 20$ KHz).}}
	\vspace{-4mm}
	\label{fig1}
\end{figure}
}
In comparison to OFDM, the main advantage of OTFS is that while Doppler spread in high mobility scenarios leads to inter-carrier interference (ICI) and therefore severe performance degradation in OFDM, in OTFS there is almost no performance degradation, as with OTFS modulation the effective DD domain channel simply shifts the transmitted information symbols $ \left( x _ { q } [ k , l ] \right) $ from one DDRE to another while preserving information \cite{OTFSOFDM, Hadaniwhitepaper, chocks, EmanueleTWC}. For example, in Fig.~\ref{fig1}, the top sub-figure illustrates the information symbols transmitted from two UTs (UT1 and UT2) in the DD domain $ ( \Delta f = 20 \mathrm { KHz } , T = 1 / \Delta f = 50 \mu s , M = N = 20) $ on shaded DDREs (slanted lines for UT1 and horizontal-vertical line grid for UT2).
The bottom sub-figure illustrates the DDREs on which information symbols have been received from each UT. We note that the information symbol transmitted from UT1 on the $(k, l)$-th DDRE is received at the BS on the $ \left( ( k - 1 ) _ { (N) } , ( l + 2 ) _ { (M) } \right) $-th DDRE\footnote{\footnotesize{For integer $x$ and positive integer $N$, $x_{(N)}$ denotes the non-negative integer less than $N$ and congruent to $x$ modulo $N$.}} due to a propagation path having a Doppler shift of $ - 1 $ KHz (which corresponds to a left circular shift by one DDRE along the Doppler domain as each DDRE is $\Delta f/N = 1$ KHz wide along the Doppler domain) and a delay of 5 $\mu s$ (which corresponds to a circular shift by two DDREs along the delay domain as each DDRE is $T/M = 2.5 \mu s$ wide along the delay domain).\footnote{\footnotesize{In this paper we consider the delay and Doppler shift of the channel paths to be general non-integer multiples of the delay domain resolution $T/M$ and Doppler domain resolution $\Delta f/N$ respectively. The delay and Doppler shifts of the channel paths for the example considered in Figs.~\ref{fig1} and \ref{fig2} are taken to be integer multiples of the delay and Doppler domain resolution respectively, only for the sake of illustrating the need for uplink timing synchronization.}}
{
\begin{figure}
\centering
	\begin{subfigure}[b]{0.8 \linewidth}
		\includegraphics[width= 4.5 in, height = 3.5 in]{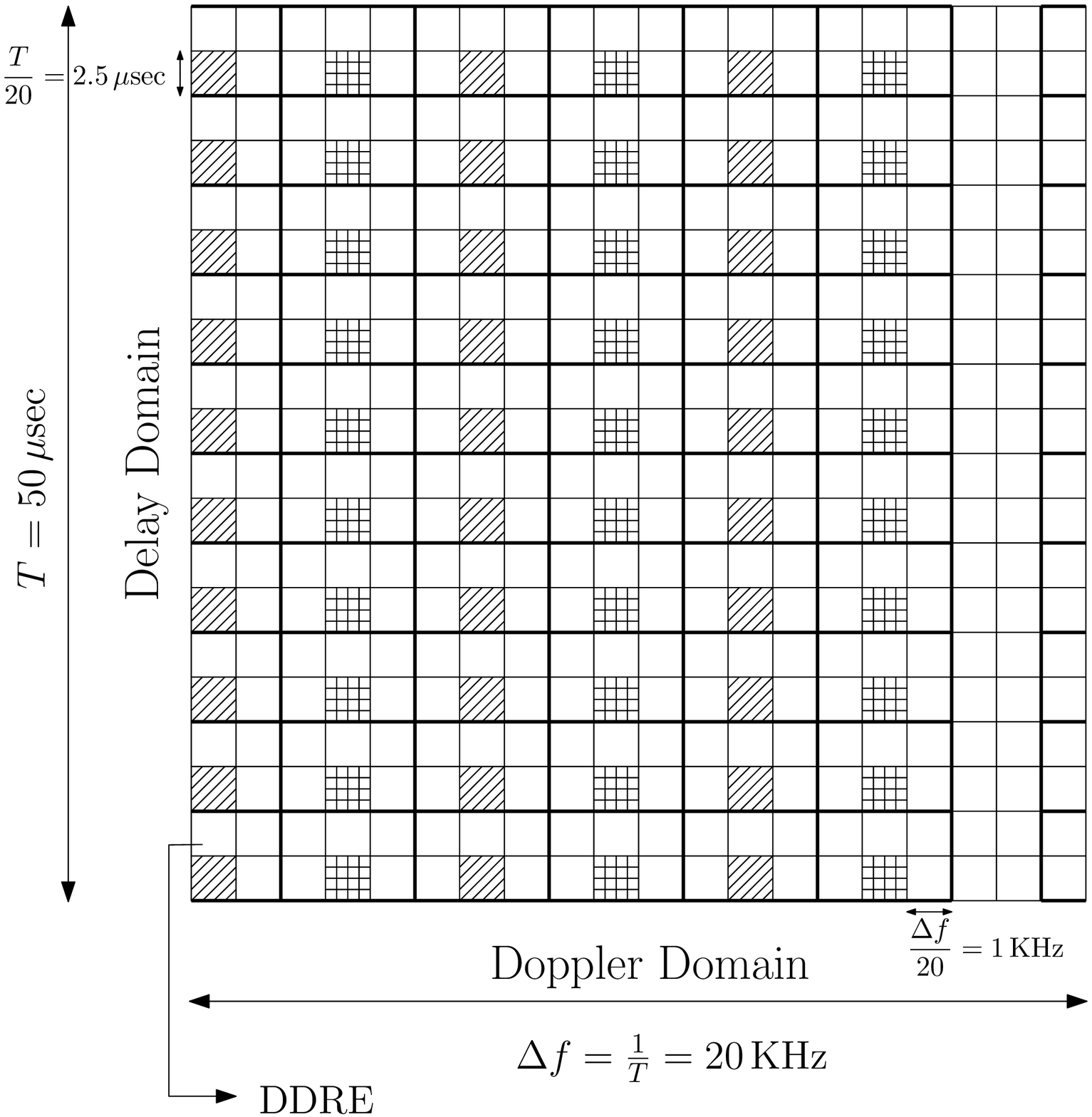}
		\caption{\centering \small{Information Tx. by two UTs in the DD domain.}}
		\label{fig2sub1}
	\end{subfigure}
	
	\vspace{4mm} 
	
	\begin{subfigure}[b]{0.8\linewidth}
		\includegraphics[width= 4.5 in, height = 3.5 in]{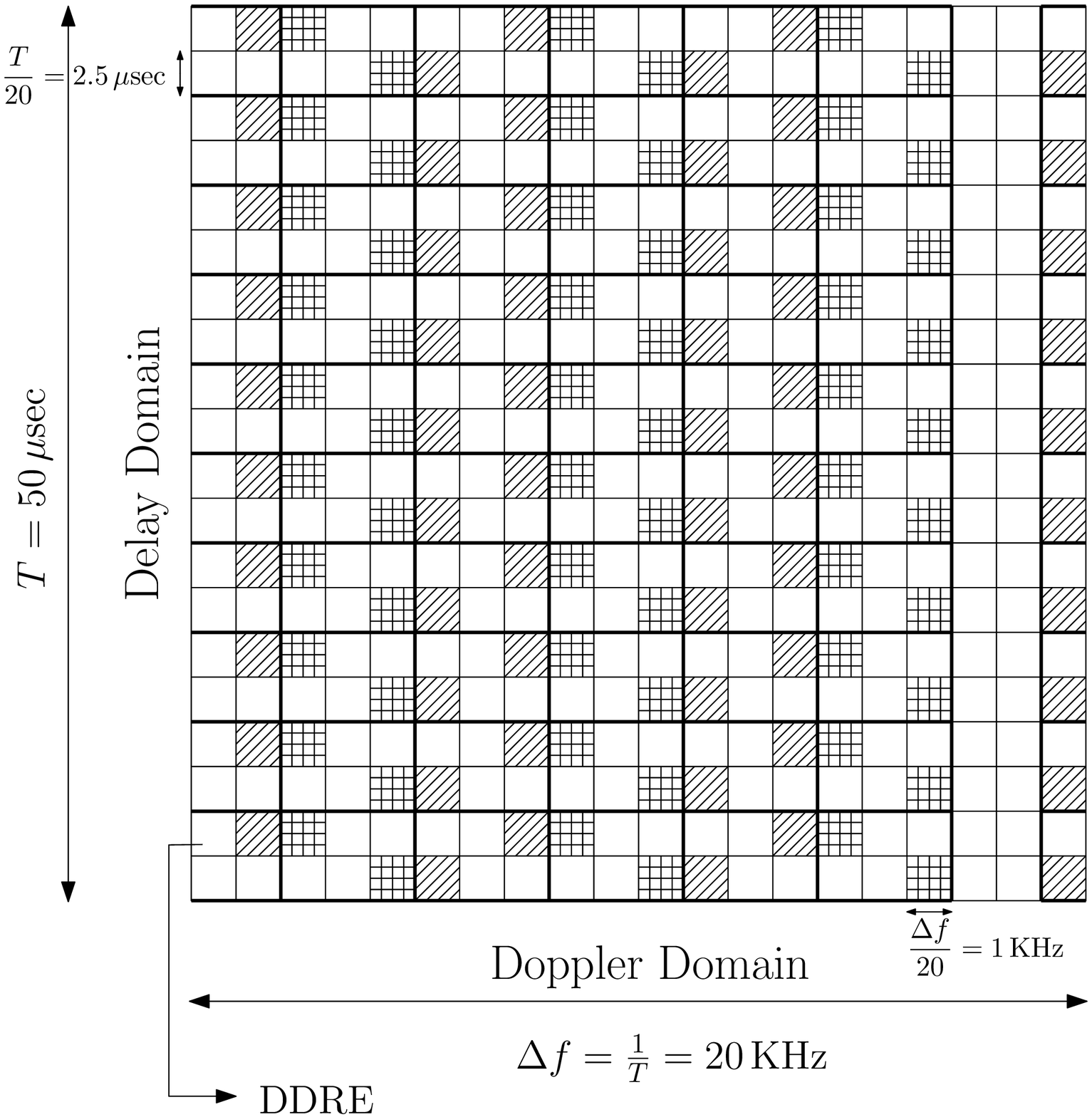}
		\caption{\centering  {\small Information Rx. at the BS in the DD domain.}}
		\label{fig2sub2}
	\end{subfigure}
	\caption{{\small With uplink timing synchronization, information symbols are transmitted on $60$ out of the total $M \times N = 400$ DDREs (see sub-figure (a)), which is twice that in Fig.~\ref{fig1sub1} in the absence of timing synchronization.}}
	\label{fig2}
	\vspace{-4mm}
\end{figure}
}

The channel induced shift in the DD domain could lead to multi-user interference (MUI) unless the UTs are properly allocated DDRE’s in such a manner that the signal received from different UTs do not overlap in the DD domain.
In the absence of uplink timing synchronization at the BS, the delay of some path from a cell-edge UT to the BS could be as high as the maximum possible round-trip propagation delay between the BS and any UT in the cell. Therefore, in order to avoid MUI, the DDRE’s allocated to different UTs must be spaced apart in the delay domain by at least the maximum round-trip propagation delay between any UT in the cell and the BS. For example, in Fig.~\ref{fig1} the channel for each UT consists of two major paths having different delays and Doppler shifts. The delay and Doppler shift profile for the two paths from UT1 to the BS is $ \{ ( 5.0 \mu s , - 1.0 \,\mathrm { KHz } ) , ( 7.5 \mu s , 1.0\, \mathrm { KHz } ) \} $ while that for UT2 is $ \{ ( 2.5 \mu s , 1.0 \,\mathrm { KHz } ) , ( 5.0 \mu s , - 1.0\, \mathrm { KHz } ) \} $. The cell has a maximum round-trip propagation delay of $ 7.5 \mu s$ which corresponds to a shift of $ 7.5 \mu s/2.5 \mu s = 3$ DDREs along the delay domain and the maximum Doppler shift magnitude of $\nu_{max} = 1 $ KHz corresponds
to a maximum shift of one DDRE (in both directions) along the Doppler domain. Therefore, each UT is allocated a contiguous rectangular block of $12$ DDREs out of which only $1$ DDRE is used to carry an information symbol and the rest $11$ DDREs serve as \textit{guard} DDREs in order to avoid MUI (in Fig.~\ref{fig1}, the rectangular blocks are demarcated by solid lines).

This method of using guard DDREs which do not carry any information has been proposed
in \cite{patent2}. However, from the example in Fig.~\ref{fig1} it is clear that in the absence of uplink timing synchronization at the BS, the guard band needed would be much larger than the actual delay spread of the channel, which would result in severe underutilization of resources.\footnote{\footnotesize{In the example, the guard band width in the delay domain is roughly equal to the maximum possible round-trip propagation delay of the cell (i.e., $7.5 \mu s$), whereas the channel delay spread for each UT is only $2.5 \mu s$.}} For example, in Fig.~\ref{fig1} only $30$ out of the $M \times N = 400$ DDREs are used for carrying information. 

On the other hand if the timing of the UTs is synchronized in the uplink, the maximum possible effective delay for any UT in the cell would be reduced to the maximum possible delay spread, which is generally much smaller than the maximum round-trip propagation delay for any UT in the cell. Due to this, the required size of the guard band would also decrease. For the same example in Fig.~\ref{fig1}, in Fig.~\ref{fig2} we show the transmitted and the received DD symbols after uplink timing synchronization. It is clear that the guard band size has now reduced from $12$ DDREs in the unsynchronized scenario of Fig.~\ref{fig1} to only $6$ DDREs in Fig.~\ref{fig2} (i.e., after uplink timing synchronization). This results in a doubling of the total system throughput as the number of information carrying DDRE’s is now increased from $30$ in Fig.~\ref{fig1} to $60$ in Fig.~\ref{fig2}.

From the above discussion it is clear that, uplink timing synchronization is a necessity for OTFS based systems also. Accurate estimation of the round-trip propagation delay (called as ``timing advance'') between a UT and the BS is therefore needed.\footnote{\footnotesize{In this paper, the round-trip propagation delay between a UT and the BS is defined to be the greatest path delay between that UT and the BS, i.e., $\tau_{q,L_q}$ for the $q$-th UT (see (\ref{equation4})).}} The central theme of this paper is that accurate estimation of the timing advance can be achieved by using OTFS modulation for the transmission of the RA preamble.  

This is because, OTFS modulation transmits symbols in the delay-Doppler (DD) domain, and the multi-path propagation simply shifts the symbols from the DDRE in which they are transmitted to other DDREs in the DD domain. Further, the amount of shift induced by a channel path along the delay domain depends only on the delay of that path and appears to be insensitive to the Doppler shift of that path. In the next section, we study this in detail and exploit it to propose OTFS based RA preamble waveforms to accurately estimate the timing advance (TA) of each UT.

\begin{figure}[h]
	\hspace{-0.2 in}
	\centering
	\includegraphics[width= 4.5 in, height= 2.0 in]{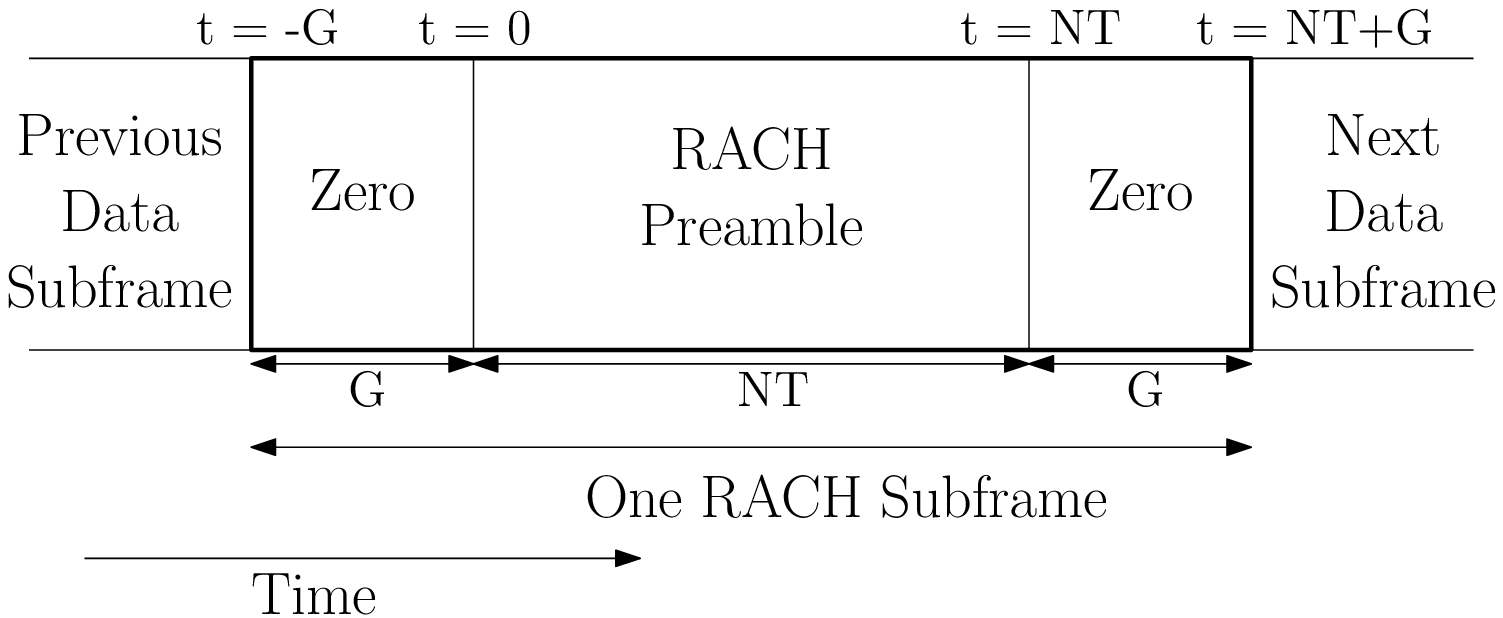}
	\vspace{-0.1 cm}
	\caption{Proposed Random Access Channel (RACH) subframe based on OTFS modulation.} 
	\vspace{-0.2cm}
	\label{fig3}
	\vspace{-0.2cm}
\end{figure}
\section{OTFS Based RA Preamble Transmission and Timing Advance (TA) Estimation}
\label{sec3}
In this section we propose the use of OTFS modulation for random access (RA) preamble transmission.
We also propose a method for RA preamble detection and TA estimation at the BS, based on the RA preambles received at the BS  in the delay-Doppler (DD) domain.
This is different from the RA waveforms and RA preamble detection/TA estimation used in 4G systems.
TA estimation in 4G systems has been explained in \cite{LTE4GTA}.
A UT randomly picks a RA preamble from the set of permissible/allowed RA preambles and transmits it.  
For each permissible RA preamble, the BS detects if that preamble has been received, and if so, it then
estimates the round-trip propagation delay between the UT which
transmitted that RA preamble and the BS. This estimate is then fed back by the BS to that UT. The UT then uses this estimate to $\textit{advance}$ its uplink timing. As each UT advances its uplink timing by its own round-trip propagation delay, it is ensured that uplink signals from all UTs would be received at the BS in a time-synchronized manner.  

As the UTs are not time synchronized during RA preamble transmission, the RA preambles transmitted from different UTs arrive at the BS in an unsynchronized manner. To remove the uncertainity of different arrival time from different unsynchronized UTs, a guard block of $G$ sec is prefixed to the RA preamble as shown in Fig.~\ref{fig3}, where $G$ sec is the maximum possible round-trip propagation delay between the BS and any UT in the cell.\footnote{\footnotesize{For the single-cell scenario considered in this paper, there might be situations where there are significant sources of reflection from objects which are outside the cell.
In such situations, the maximum possible round-trip propagation delay between the BS and any UT in the cell (i.e., $G$) must also include the effect of propagation paths due to reflection from such objects.}} Also, a guard block of $G$ seconds is suffixed after the transmission of $s_q(t)$ in order that the RA preamble transmission does not interfere with subsequent frames (see Fig.~\ref{fig3}). Even in 4G systems, a guard block of duration equal to the maximum possible round-trip propagation delay is prefixed and suffixed to the RA preamble \cite{4GLTE}. Let $T_c >  G$ be the time allocated for the RA
preamble (excluding the prefix and suffix guard blocks of $G$ sec each) and $B_c$ denote the total bandwidth allocated for RA preamble transmission.
Let the entire DD-domain consisting of the $MN$ DDREs be denoted by

{\vspace{-3mm}
\small 
\begin{eqnarray}
{\mathcal D}  \Define  \{ (k,l) \, | \, k = 0,1,\cdots, (N-1) \,,\, l = 0,1, \cdots, (M-1) \}. 
\end{eqnarray}
\normalsize}
Note that the Doppler domain and delay domain indices of the $(k,l)$-th DDRE i.e., $k$ and $l$ respectively, are both non-negative integers.
In the proposed OTFS based RA preamble waveform, the $q$-th RA preamble is allocated a subset ${\mathcal R}_q$ of ${\mathcal D}$. Further,
these subsets (referred to as DDRE groups) form a partition of the entire DD-domain i.e.,
\begin{eqnarray}
\cup_{i=0}^{R-1} {\mathcal R}_i \, = \, {\mathcal D}  \,\,\,,\,\,\,   {\mathcal R}_i \cap {\mathcal R}_k \, = \, \emptyset  \,\,\, (i \ne k),  \nonumber \\
 i,k=0,1,\cdots, (R-1) 
\end{eqnarray}
where $R$ is the total number of allowed RA preambles.
Further, the $q$-th RA preamble waveform consists of a non-zero symbol on a single DDRE location $(k_q,l_q) \in {\mathcal R}_q$ and zeros on all other DDREs, i.e.,
the DD-domain transmit signal for the $q$-th RA preamble is given by
\begin{eqnarray}
\label{equation31}
x _ { q } [ k,l ] = \begin{cases} 
{ \sqrt { M N E}  } & { , \mbox{if}  \,\,\, \, (k,l) = (k_q, l _ { q }) } \\
 { 0 } & {  ,  \mbox{if}  \,\,\,\, \, (k,l) \ne  (k_q,  l _ { q })  } 
 \end{cases}
\end{eqnarray}
where $E$ denotes the energy of the transmitted RA preamble.
As a UT can only transmit one RA preamble at a time, from (\ref{equation17}), (\ref{equation18}) and (\ref{gtxpulse_eqn}) it follows that the time-domain transmit signal for the $q$-th RA preamble is given by (\ref{sqprmb}) (see top of this page).
\begin{figure*}
{\vspace{-7mm}
\small
\begin{eqnarray}
\label{sqprmb}
s_q(t)  =   \begin{cases}
0  & , \mbox{if}  \,\,\,\,  -G \leq t < 0  \\
\sqrt{\frac{E}{MN}}  \, \sum_{m=0}^{M-1} \sum_{n=0}^{N-1} e^{j2 \pi m \Delta f t} e^{\frac{j 2 \pi n k_q}{N}} e^{\frac{-j 2 \pi m l_q}{M}}  g_{tx}(t - nT) & ,  \mbox{if}  \,\,\,\, 0 \leq t < NT \\
0 & ,  \mbox{if}  \,\,\,\, NT \leq t < NT+G
\end{cases}. \\
\hline \nonumber
\end{eqnarray}
\vspace{-5mm}
\normalsize}
\end{figure*}
As a rectangular transmit pulse is used (see (\ref{gtxpulse_eqn})), from (\ref{sqprmb}) it follows that $\int_{0}^{NT}  \vert s_q(t) \vert^2 dt = E$, i.e., the transmitted RA preamble energy is $E$. With a rectangular transmit pulse, the transmitted time-domain signal for $t \in [nT  \,,\,  (n+1)T)$ is

{\vspace{-3mm}
\small
\begin{eqnarray}
\label{sqprmb1}
s_q(t)    =    \sqrt{\frac{E}{MN T}}  \, e^{j 2 \pi n k_q/N}  \sum_{m=0}^{M-1}e^{j2 \pi m \Delta f t} e^{-j 2 \pi m l_q/M}.
\end{eqnarray} 
\normalsize}
From (\ref{sqprmb}), it follows that during the time interval $[0 \,,\, NT)$ the average radiated power is $E/NT$, and therefore the peak-to-average-power-ratio (PAPR) is given by

{\vspace{-4mm}
\small
\begin{eqnarray}
\label{eqnpapr1}
\mbox{\small{PAPR}} & \hspace{-3mm} = &  \hspace{-3mm} \max_{t \in [0 \,,\, NT)} \frac{\vert s_q(t) \vert^2}{E/NT} \nonumber \\
 & \hspace{-20mm} \mya & \hspace{-14mm} \max_{n = 0,1,\cdots, N-1} \left[ \max_{t \in [nT \,,\, (n+1)T)}  \hspace{-3mm} \frac{\vert   \sum_{m=0}^{M-1}e^{j2 \pi m \Delta f t} e^{-j 2 \pi m l_q/M}  \vert^2}{M} \right] \nonumber \\
 & \hspace{-20mm} = & \hspace{-14mm} \max_{n = 0,1,\cdots, N-1} \left[ \max_{t \in [0 \,,\, T)}  \hspace{-1mm} \frac{\vert   \sum_{m=0}^{M-1}e^{j2 \pi m \Delta f (t+nT)} e^{-j 2 \pi m l_q/M}  \vert^2}{M} \right] \nonumber \\
 & \hspace{-20mm} \myb & \hspace{-14mm} \max_{n = 0,1,\cdots, N-1} \left[ \max_{t \in [0 \,,\, T)}  \hspace{-1mm}   \frac{\vert   \sum_{m=0}^{M-1}e^{j2 \pi m \Delta f t} e^{-j 2 \pi m l_q/M}  \vert^2}{M}   \right] \nonumber \\
  & \hspace{-20mm} \myc & \hspace{-11mm} \max_{t \in [0 \,,\, T)}  \hspace{-1mm}  \frac{\vert   \sum_{m=0}^{M-1}e^{j2 \pi m (t \Delta f -   l_q/M )} \vert^2}{M}  \, = \, M
\end{eqnarray}
\normalsize}
where step (a) follows by substituting the expression for $s_q(t)$ from (\ref{sqprmb1}) into the first step above.
Step (b) follows from the fact that $T \Delta f  = 1$ due to which $e^{j 2 \pi m \Delta f n T} = e^{j 2 \pi m n} = 1$.  
Using the identity, $\left\vert \sum\limits_{m=0}^{M-1} a_m \right\vert^2 \leq M \sum\limits_{m=0}^{M-1} \vert a_m \vert^2, a_m \in {\mathbb C}$,
with $a_m = e^{j 2 \pi m \left( t \Delta f  - l_q/M \right)}, m=0,1,\cdots,M-1,$ we get $ \frac{\vert   \sum_{m=0}^{M-1}e^{j2 \pi m (t \Delta f -  l_q/M )} \vert^2}{M}  \leq M$.
Since $l_q \in \{0 , 1, \cdots, M-1 \}$, $l_q T/M \in [0 \,,\, T)$, the maximum value in step (c) is achieved at $t = l_q T/M$.   
Therefore, the PAPR for each RA preamble is $M$ (i.e., it does not depend on $N$).
Clearly, $M$ should be as small as possible in order to limit PAPR.
As we shall see later, the proposed timing advance estimate is based on the delay domain index of the DDRE where the maximum RA preamble energy is received. Since the delay domain is $T$ sec wide and is divided/discretized into $M$ equal sub-divisions with each sub-division being $T/M$ sec wide, the timing advance estimate for any received RA preamble
can only take the values $0, \frac{T}{M}, \frac{2T}{M}, \cdots, \frac{(M-1)T}{M}$. As the maximum possible value of the timing advance estimate is $(M-1)T/M$, the timing advance estimation error can be very high for any UT for which the round-trip propagation delay between the BS and that UT exceeds $(M-1)T/M$. 
Therefore we should choose $(T,M)$ in such a way that even the maximum possible round-trip propagation delay $G$ can be measured accurately, i.e.
\begin{eqnarray}
\label{gtmeqn}
\frac{(M-1) T}{M} & > & G.
\end{eqnarray}
From (\ref{equation18}) it is clear that the transmitted time-domain signal $s_q(t)$ corresponding to the transmit DD domain signal $x_q[k,l]$ has a bandwidth of $M \Delta f = M/T$ Hz. Since the bandwidth constraint on the transmitted RA preamble waveforms is $B_c$, it follows that $(M,T)$ must also satisfy
\begin{eqnarray}
\label{bndcnstr}
{M}/{T} & \leq & B_c.
\end{eqnarray}
As the TA estimation accuracy is limited to half the width $T/M$ of each DDRE along the delay domain, from (\ref{bndcnstr}) it follows that the TA estimation accuracy is limited
to $1/(2B_c)$. To achieve the best possible accuracy for a given $(B_c,T_c)$ while ensuring the constraints in (\ref{gtmeqn}) and (\ref{bndcnstr}) and also that the PAPR of the transmit signal is as small as possible, we choose $(M,T)$ to be
\begin{eqnarray}
\label{mtvaleqn}
M & = &  1 + \lceil G B_c \rceil \,\,,\,\, T   =  (1 + \lceil G B_c \rceil)/B_c.
\end{eqnarray}        
Note that choosing an $M$ larger than $1 + \lceil G B_c \rceil $ will not improve the TA estimation accuracy beyond $1/(2B_c)$ but will increase the PAPR (see (\ref{eqnpapr1})).\footnote{\footnotesize{
Since the maximum possible round-trip propagation delay is $G$, multiple RA preambles can be multiplexed along the delay domain
only if the delay domain width $T$ is much larger than $G$ (i.e., $T \gg G$). Since $M/T$ is the bandwidth of the RA preamble waveform, for a given bandwidth $B_c$, we have $M = T B_c \gg G B_c$, i.e., the number of delay domain sub-divisions $M$ will have to be much larger than the value in (\ref{mtvaleqn}), which will result in high PAPR.}}
Since the transmit pulse $g_{tx}(t)$ is time-limited to $T$ seconds, from (\ref{equation18}) it is clear that the transmitted time-domain signal $s_q(t)$ is limited to $NT$ seconds.
Since the total time constraint on the transmit signal is $T_c$, $N$ must satisfy

{\vspace{-5mm}
\small 
\begin{eqnarray}
\label{nvalbnd}
N T & \leq & T_c.
\end{eqnarray}
\normalsize}
Next, we propose the DDRE group ${\mathcal R}_q, q = 0,1,\cdots, (R-1)$ corresponding to the $q$-th RA preamble to be
\begin{eqnarray}
\label{ddredef}
{\mathcal R_q} & \Define & \{  (k,l)  \, \vert \, 0 \leq l \leq (M-1) \,,\, qN_1  \leq k  < (q+1) N_1 \} \nonumber \\
\end{eqnarray}
where $N_1$ is the width of each RA preamble DDRE group along the Doppler domain.
The number of allowed RA preambles is therefore given by
\begin{eqnarray}
\label{reqn}
R & = & \left\lfloor {N}/{N_1} \right\rfloor.
\end{eqnarray}
From the proposed definition of the DDRE groups for the RA preambles in (\ref{ddredef}), it is clear that the RA preambles
are allocated different non-overlapping groups along the Doppler domain. The energy of each transmitted RA preamble is localized to a particular DDRE
in its corresponding group, i.e., the energy of the $q$-th RA preamble is localized at the $(k_q,l_q)$-th DDRE ($(k_q,l_q) \in {\mathcal R}_q$).   
Due to channel induced Doppler shift, the Doppler domain index of the DDRE where energy of a RA preamble is received need not be same as
the Doppler domain index of the DDRE where the energy of the corresponding transmitted RA preamble is localized. If some energy of an RA preamble is received outside its
DDRE group, this can lead to interference between different RA preambles.     
In order to minimize this interference, the respective DDRE locations of energy localization for two adjacent RA preambles (e.g., $q$-th and $(q+1)$-th RA preambles) are separated by $N_1$ DDREs along the Doppler domain, i.e.
\begin{eqnarray}
(k_q , l_q) & \in & {\mathcal R}_q  \,\,,\,\,   k_0  = \left\lfloor N_1/2  \right\rfloor  \,,\, q=0,1,\cdots, (R-1) \nonumber \\
(k_{q+1} - k_{q}) & = & N_1 \,\,\,,\,\,\, q=0,1,\cdots, (R-2).
\end{eqnarray}
If the Doppler shifts are integer multiple of the Doppler domain resolution $\Delta f/N$, then having $N_1 = \frac{2 \nu_{max}}{\Delta f/N} + 1$ would ensure
zero interference between different RA preambles. In this paper, we consider Doppler shifts which are in general non-integer multiples of $\Delta f/N$, due to which it is not possible
to ensure zero interference. However, receiver windowing along the time-domain allows us to achieve a trade-off between $N_1$ and the interference between different received RA preambles. This trade-off is discussed in more detail in Section \ref{secmultiUT}.    
In the event that two or more UTs randomly transmit the same RA preamble (i.e., collision event), it is clear that even in the absence of additive noise the timing advance of only one UT (whose RA preamble has been received at the BS with the highest received power) is estimated properly. Hence, for a low average timing error probability (TEP), the number of RA preambles $R = \lfloor N/N_1  \rfloor$ must be as large as possible in order to minimize the probability of collisions. Hence, we must choose $N$ to be the largest integer which satisfies (\ref{nvalbnd}), i.e.

{\vspace{-3mm}
\small
\begin{eqnarray}
\label{Neqn1}
N  =  \left\lfloor \frac{T_c}{T} \right\rfloor  =  \left\lfloor \frac{B_c T_c}{1 + \lceil G B_c \rceil } \right\rfloor. 
\end{eqnarray}
\normalsize}
In Fig.~\ref{figrarq}, we illustrate the DDRE groups ${\mathcal R}_q, q =0,1,\cdots, (R-1)$ for
$B_c = 1.08$ MHz, $T_c = 1.6$ ms, $G = 15\mu s$, $\nu_{max} = 300$ Hz. Based on (\ref{mtvaleqn}) we have $M = 18$, $T = M/B_c = 16.66 \mu s$, $\Delta f = 1/T = 60$ KHz and $N = 96$ (see (\ref{Neqn1})).
With $N_1 = 5$ we then have $R = \lfloor N/N_1 \rfloor = 19$.
From (\ref{sqprmb}) it follows that the total transmitted energy of an RA preamble is given by  
\begin{eqnarray}
\label{equation33}
\int _ { 0 } ^ { NT } \left| s _ { q } ( t ) \right| ^ { 2 } d t = E \,\,\,,\,\, q=0,1,\cdots, (R-1).
\end{eqnarray}
Let $c_q \,,\, q=0,1,\cdots, (Q-1)$ be one when the $q$-th UT transmits a RA preamble and zero otherwise. Also, let $r_q, q \in \{0 , 1, \cdots, (Q-1)  \}$ denote the index of the RA preamble transmitted by the $q$-th UT.
From (\ref{equation5}) and (\ref{sqprmb}), the signal received at the BS is then given by (\ref{rxraprmb}) (see top of this page).
\begin{figure*}
{\vspace{-8mm}
\small
\begin{eqnarray}
\label{rxraprmb}
r(t) & = & \sum\limits_{q=0}^{Q-1} \sum\limits_{i=1}^{L_q} c_q h_{q,i} \, {s}_{r_q}(t - \tau_{q,i}) e^{j 2 \pi \nu_{q,i} (t - \tau_{q,i})}  \, +  w(t) \nonumber \\
& \hspace{-16mm} =  & \hspace{-9mm} \sqrt{\frac{E}{MN}} \sum_{q=0}^{Q-1} \sum_{i=1}^{L_q} \sum_{m'=0}^{M-1} \sum_{n'=0}^{N-1} c_q  h_{q,i}  e^{j 2 \pi \nu_{q,i} (t - \tau_{q,i})} e^{j 2 \pi m' \Delta f (t - \tau_{q,i})} e^{\frac{j 2 \pi n' k_{r_q}}{N}} e^{\frac{-j 2 \pi m' l_{r_q}}{M}} g_{tx}(t - n'T - \tau_{q,i})  \, +  \, w(t).
\end{eqnarray}
\normalsize}
\end{figure*}
From (\ref{equation2}) and (\ref{equation33}) it follows that the average transmit power of an RA preamble
at the BS is $E/(NT+ 2G)$. The noise power is $M \Delta f N_o$ Watt
where $N_o$ Watt/Hz is the power spectral density of the AWGN $w(t)$ in (\ref{rxraprmb}). Hence the transmit signal-to-noise-ratio of a RA preamble is

{\vspace{-3mm}
\small
\begin{eqnarray}
\label{equation34}
\rho \Define \frac { E / ( NT+ 2G ) } { M N _ { o }  \Delta f} = \frac { E } { N _ { o } } \frac { 1 } { M \left( N +  \frac{2GB_c}{1 + \lceil G B_c \rceil } \right) }.
\end{eqnarray}
\normalsize}
The BS ignores the first $G$ seconds of $r(t)$ due to the uncertainty in the arrival time of RA preambles from
the unsynchronized UTs. Next, using (\ref{equation20}), the received time-frequency (TF) signal at the BS is given by (\ref{Ynmeqn1}) (see top of this page, here $g_{rx}(t) = g_{tx}(t)$).
\begin{figure*}
{\vspace{-9mm}
\small
\begin{eqnarray}
\label{Ynmeqn1}
Y[n,m] & = & W_{rx}[n,m] \int\limits_{0}^{NT}  g_{rx}^*(t-nT) r(t) e^{-j 2 \pi m \Delta f t}  dt
\, = \, \frac{W_{rx}[n,m] }{\sqrt{T}} \int\limits_{nT}^{(n+1)T}  r(t) e^{-j 2 \pi m \Delta f t}  dt  \nonumber \\
&  \mya  &   \sqrt{\frac{E}{MN}}  W_{rx}[n,m]  \sum\limits_{q=0}^{Q-1} \sum\limits_{i=1}^{L_q} \sum_{m'=0}^{M-1} {\Bigg [} c_q \, h_{q,i}  e^{-j 2 \pi \nu_{q,i} \tau_{q,i}}  e^{-j 2 \pi m' \Delta f \tau_{q,i} }  e^{\frac{j 2 \pi n k_{r_q}}{N}} e^{\frac{-j 2 \pi m' l_{r_q}}{M}} \nonumber \\
&    &  \hspace{6mm} {\Bigg \{}  \frac{1}{T} \int_{nT + \tau_{q,i}}^{(n+1)T} e^{j 2 \pi (\nu_{q,i} + (m' - m) \Delta f)t}  dt  \, + \,  \frac{e^{\frac{- j 2 \pi k_{r_q}}{N}}}{T}  \int_{nT}^{nT + \tau_{q,i}} e^{j 2 \pi (\nu_{q,i} + (m' - m) \Delta f)t}  dt   {\Bigg \}} \, {\Bigg ]}  \, + \, W[n,m], \nonumber \\
W[n,m]  & \Define &  \frac{W_{rx}[n,m] }{\sqrt{T}} \int\limits_{nT}^{(n+1)T} w(t) e^{-j 2 \pi m \Delta f t} dt. 
\end{eqnarray}
\vspace{-4mm}
\normalsize}
\end{figure*}
In (\ref{Ynmeqn1}), step (a) follows from the fact that $g_{tx}(\cdot)$ is rectangular (see ({\ref{gtxpulse_eqn}})). Further, from (\ref{Ynmeqn1})
it also follows that $W[n,m] \sim  {\mathcal C}{\mathcal N}(0, \vert W_{rx}[n,m] \vert^2 N_o)$. Since the proposed DD domain RA preambles interfere only along the Doppler domain, receiver windowing along the time-domain is sufficient to reduce leakage of RA preamble energy from one RA preamble to an adjacent RA preamble along the Doppler domain. Therefore, in this paper we replace $W_{rx}[n,m]$ by $W_{rx}[n]$. The TF signal $Y[n,m]$ is then transformed back to the DD domain using the SFFT transform as in (\ref{equation21}). The resulting DD signal ${\widehat x}[k,l] $ is given by (\ref{xkleqn}) (see end of this page).
\begin{figure*}
{\vspace{-6mm}
\small
\begin{eqnarray}
\label{xkleqn}
{\widehat x}[k,l]      &    = &   \sum_{m=0}^{M-1} \sum_{n=0}^{N-1}   Y[n,m] \, e^{j 2 \pi    (\frac{ml}{M} - \frac{nk}{N})  }   \,\,,\,\, k =0,1,\cdots, (N-1), l=0,1,\cdots, (M-1) \nonumber \\
&   \mya &   \sqrt{E M N} \sum\limits_{q=0}^{Q-1} \sum\limits_{i=1}^{L_q} {\Bigg \{} c_q h_{q,i} e^{-j 2 \pi \nu_{q,i} \tau_{q,i} } U_{q,i}[k]  V_{q,i}[l]    {\Bigg \}}  \, + \, w[k,l] \,,\,
w[k,l]  \Define   \sum_{m=0}^{M-1} \sum_{n=0}^{N-1}   W[n,m] \, e^{j 2 \pi    (\frac{ml}{M} - \frac{nk}{N})  } 
\end{eqnarray}
\normalsize}
\end{figure*}
\begin{figure*}
{\vspace{-4mm}
\small
\begin{eqnarray}
\label{xkleqndef}
U_{q,i}[k] & \Define &  \frac{1}{N}  \sum\limits_{n=0}^{N-1} W_{rx}[n] e^{j 2 \pi n \left( \frac{(k_{r_q} - k )}{N} + \frac{\nu_{q,i}}{\Delta f}  \right)}  \nonumber \\
V_{q,i}[l]  & \Define & e^{j \pi (M-1) \left( \frac{l - l_{r_q}}{M}  -  \frac{\tau_{q,i}}{T}\right)} 
 {\Bigg \{}  \int\limits_{\frac{\tau_{q,i}}{T}}^{1} e^{j 2 \pi \frac{\nu_{q,i}}{\Delta f} t} \, \frac{M \mbox{Sinc}\left(M \left( t - \frac{l_{r_q}}{M} - \frac{\tau_{q,i}}{T} \right) \right)}{\mbox{Sinc}\left( t - \frac{l_{r_q}}{M} - \frac{\tau_{q,i}}{T} \right)}  \, \frac{ \mbox{Sinc}\left(M \left( t - \frac{l}{M} \right) \right)}{\mbox{Sinc}\left( t - \frac{l}{M} \right)}   \, dt    \nonumber \\
& &  \hspace{33mm} + \, e^{\frac{- j 2 \pi k_{r_q} }{N}}  \int\limits_{0}^{\frac{\tau_{q,i}}{T}}  e^{j 2 \pi \frac{\nu_{q,i}}{\Delta f} t} \, \frac{M \mbox{Sinc}\left(M \left( t - \frac{l_{r_q}}{M} - \frac{\tau_{q,i}}{T} \right) \right)}{\mbox{Sinc}\left( t - \frac{l_{r_q}}{M} - \frac{\tau_{q,i}}{T} \right)}  \, \frac{ \mbox{Sinc}\left(M \left( t - \frac{l}{M} \right) \right)}{\mbox{Sinc}\left( t - \frac{l}{M} \right)}   \, dt       {\Bigg \}}. \\
\hline \nonumber
\end{eqnarray}
\vspace{-4mm}
\normalsize}
\end{figure*}
In (\ref{xkleqn}), step (a) follows from (\ref{Ynmeqn1}). 
For detecting the $q$-th RA preamble, the BS computes

{\vspace{-5mm}
\small
\begin{eqnarray}
\label{equation36}
z _ { q } \Define \max _ {  (k,l)  \in \mathcal { R } _ { q } } | \widehat { x } [ k,  l ] | ^ { 2 }.
\end{eqnarray}
\normalsize}
The TA estimate for the $q$-th RA preamble (normalized by $T/M$) is given by

{\vspace{-5mm}
\small
\begin{eqnarray}
\label{equation38}
\widehat { \mathrm { TA } } _ { q } & \Define &  \begin{cases}
 T_q[k,l]  & ,   \mbox{if}  \,\,\,\,  T_q[k,l]  \, \geq \, 0  \\
M +  T_q[k,l]  & ,     \mbox{otherwise} 
\end{cases} \nonumber \\
T_q[k,l]  &  \Define &    \underset { l  \, |  \, (k,l) \in \mathcal{R}_q  } { \operatorname { argmax } } | \widehat { x } [ k, l ] | ^ { 2 }  - l _ { q }.
\end{eqnarray}
\normalsize}
Therefore, for the $q$-th RA preamble, the proposed normalized TA estimate is based on the delay domain index of the DDRE (in the $q$-th RA preamble DDRE group ${\mathcal R}_q$) where the maximum
energy is received, i.e., the TA estimate (normalized by $T/M$) takes values in the discrete set $\{ 0 , 1, 2, \cdots, M-1 \}$. Therefore, the TA estimation accuracy is limited by half of the delay domain resolution i.e., $T/2M$. From (\ref{mtvaleqn}), $\frac{T}{2M} = \frac{1}{2B_c}$ and therefore the TA estimation accuracy is
limited by half of the inverse bandwidth of the RA preamble waveform.\footnote{\footnotesize{This limit is due to the quantization of the delay domain into $M$ equal sub-divisions of
width $T/M$. Although the magnitude of the worst case delay quantization error is $1/(2 B_c)$, the root-mean square quantization error will be $\frac{1}{\sqrt{12}} \frac{T}{M} = \frac{1}{\sqrt{12} B_c}$ if the path delays are equally likely to lie anywhere within a sub-division (i.e., uniform distribution within the sub-division).}} This limit can be reduced by increasing the bandwidth $B_c$.

Let the $q$-th RA preamble be transmitted by the $q'$-UT. In this paper, the $q$-th RA preamble is said to be {\em detected} if and only if,
i) the TA estimate $\widehat { \mathrm { TA } }_{ q }$ lies between the normalized and quantized smallest round-trip delay $\lfloor M\tau_{q',1}/T \rfloor $ (corresponding to the first path) and the normalized and quantized greatest round-trip delay $\lceil M \tau_{q',L_{q'}}/T \rceil $ (corresponding to the $L_{q'}$-th path), and
ii) the maximum energy $z_q$ received in the DDRE group ${\mathcal R}_q$ exceeds a pre-determined threshold $\mu$.  
This threshold is chosen so as to achieve a desired probability of false alarm which is given by

{\vspace{-5mm}
\small
\begin{eqnarray}
\label{equation37}
P_{ F A } \hspace{2mm}\Define \hspace{2mm}\mbox{Pr}\left( z _ { q } \geq \mu \,| \,q \text {-th RA Preamble is not Tx } \right).
\end{eqnarray}
\normalsize}
Therefore for a given desired false alarm probability $P_{FA}$, the probability of missed detection of the $q$-th RA preamble is given by

{\vspace{-5mm}
\small
\begin{eqnarray}
\label{equation39}
P _ { e } \Define \operatorname { Pr } {\Bigg (} \widehat { \mathrm { TA } } _ { q } \notin \left[  \left\lfloor M \tau _ { q' , 1 } / T \right\rfloor ,  \left\lceil M \tau _ { q' , L _ { q' } } / T  \right \rceil \right]  \text { or }  z _ { q } < \mu  \,\Big|\, q \hspace{-1mm}\text { -th } \mathrm { RA } \text { Preamble is } \mathrm { Tx }  \, \text { by  the}  \, q' \hspace{-1mm}\text { -th }  \, \mathrm { UT }  {\Bigg )},
\end{eqnarray}
\normalsize}
where $\mu$ is the threshold chosen to achieve the desired false alarm probability $P_{FA}$. Subsequently in this paper, we refer to this missed detection event as the timing error event and the missed detection probability as the timing error probability (TEP).

For a given UT, the delay of the channel path having the highest channel power gain can be anywhere between the smallest and the greatest path delay between that UT and the BS.
Therefore, even in a noise-free and interference-free scenario, the TA estimation error $(\tau_{q',L_{q'}} - \frac{T}{M}\widehat { \mathrm { TA } } _ { q })$ will be of the order of the channel delay spread between the $q'$-UT
(which transmitted the $q$-th RA preamble) and the BS.   
Due to this, consecutive data frame transmissions (where information is transmitted) are separated by a guard time-interval of duration roughly equal to the maximum possible channel delay spread between any
UT in the cell and the BS.\footnote{\footnotesize{This is also true in OFDM systems where a cyclic-prefix of duration equal to the maximum possible channel delay spread is used to avoid
interference between consecutive OFDM symbols.}}
From (\ref{equation39}) it follows that a successful detection (i.e., no timing error event) implies that the
TA estimation error is positive and is less than the channel delay spread.
Further, a positive TA estimation error which is less than the channel delay spread will not lead to any interference between consecutive data frame transmissions due
to the guard time-interval. In other words, a successful detection is representative of the fact that the TA estimation error is small enough
such that there is no inter-frame interference.

It is expected that in an ideal noise-free and interference-free scenario, the
proposed TA estimate will be equal to the normalized delay of the path having the highest channel power gain.
However, if the timing error event in (\ref{equation39}) were to be re-defined
to occur when the proposed TA estimate is not equal to the normalized delay of the path having the highest channel power gain, then    
a timing error event could happen even when the proposed TA estimate lies between the smallest and
the greatest path delays (i.e., even when the TA estimation error is positive and is less than the channel delay spread). Since there is anyways a guard time-interval between consecutive data frame transmissions, such a definition for timing error events will result in unnecessary random access (RA) call failures which will increase the effective latency of the RA procedure. For this reason we consider the definition
of the timing error event to be as in (\ref{equation39}).

When the UTs have i.i.d. fading statistics and their locations inside the cell are also i.i.d., then the TEP averaged over the fading and location statistics is the same for each UT, and is given by

{\vspace{-4mm}
\small
\begin{eqnarray}
\label{equation323}
P_e & =  &  \sum\limits_{k=1}^{\infty}   \mbox{Pr}(\mbox{Timing error} \, , \, Q = k \, \vert \, Q \geq 1 )  \nonumber \\
& = &  \mbox{Pr}(Q = 1 \, \vert \, Q \geq 1 )  \mbox{Pr}(\mbox{Timing error} \, | \, Q = 1)  \,  +  \, \sum\limits_{k=2}^{\infty} \mbox{Pr}(Q = k \, \vert \, Q \geq 1 )  \mbox{Pr}(\mbox{Timing error} \, | \, Q = k) 
\end{eqnarray}
\normalsize}
where $Q$ is the random number of UTs transmitting a RA preamble. Also, the conditioning over the event $\{ Q \geq 1 \}$ is because in (\ref{equation39}) the probability of timing error is conditioned on the fact that the RA preamble has been transmitted.
\subsection{Single UT Scenario}
\label{secsingleUT}
In (\ref{equation323}), the term $\mbox{Pr}(Q = 1 \, \vert \, Q \geq 1 ) \mbox{Pr}(\mbox{Timing error} \, | \, Q = 1)$ is the TEP for the single UT scenario.
When only a single UT (say $q$-th UT) transmits an RA preamble, a timing error event can only happen due to AWGN/fading/Doppler. For a given desired false alarm probability (see (\ref{equation37})), the threshold $\mu$ needs to be sufficiently high. Therefore, at low SNR $\rho$, it is possible that the received RA preamble power could fall below the threshold $\mu$ leading to missed-detection. However, at high SNR, the only possible cause of a timing error event is due to the channel induced Doppler shift.
Therefore in the following we analyze the impact of Doppler shift on the accuracy of the proposed TA estimate. From the expression for ${\widehat x}[k,l]$ in (\ref{xkleqn}) we note that the R.H.S. depends on $k$ and $l$ through two different terms. The term $U_{q,i}[k]$ depends only on $k$ and $V_{q,i}[l]$ depends only on $l$. In the following we show why this separation of terms makes the proposed TA estimate almost \emph{invariant} of the Doppler shift. 

From (\ref{equation38}), it is clear that the proposed TA estimate $\widehat { \mathrm { TA } } _ { r_q }$ of the $r_q$-th RA preamble transmitted by the $q$-th UT, depends \emph{only} on the index $l'$ of the highest energy DDRE symbol ${\widehat x} [k',l']$ over all ${\widehat x}[k,l]$, $(k,l) \in {\mathcal R}_{r_q}$. Due to the separation of terms depending on $l$ and $k$ in (\ref{xkleqn}), the highest energy DDRE symbol ${\widehat x}[k',l']$ must have its delay domain index $l'$ such that for some $i'$, $V_{q,i'}[l']$ has the largest magnitude over all $V_{q,i}[l]$, $i=1,2 \cdots, L_q$, $l=0,1,\cdots, (M-1)$. In the expression for $V_{q,i}[l]$ in (\ref{xkleqndef}) (see top of previous page), 
the term $\frac{M Sinc\left(M \left( t - \frac{l_{r_q}}{M} - \frac{\tau_{q,i}}{T} \right) \right)}{Sinc\left( t - \frac{l_{r_q}}{M} - \frac{\tau_{q,i}}{T} \right)}$ inside the integral has a sharp peak at $t =   l_{r_q}/M +  \tau_{q,i}/T$ whereas the other term $\frac{ Sinc\left(M \left( t - \frac{l}{M} \right) \right)}{Sinc\left( t - \frac{l}{M} \right)}$ has a peak at $t = l/M$. Hence $V_{q,i}[l]$ will have a large magnitude only when these peaks overlap i.e., when $ l/M \approx l_{r_q}/M +  \tau_{q,i}/T $. Therefore, for the single UT scenario (say $q$-th UT) at high SNR, the peak magnitude of ${\widehat x}[k,l]$ is expected to be at the delay domain index $l = l_{r_q} + \lfloor M \tau_{q,i}/T \rfloor$ or $l =  l_{r_q} + \lceil M \tau_{q,i}/T \rceil$ for some $i=1,2,\cdots,L_q$. Hence $\widehat { \mathrm { TA } } _ { r_q }$ is expected to lie in the set $\{ \lfloor M \tau_{q,1}/T \rfloor  , \cdots,  \lceil M \tau_{q,L_q}/T \rceil \}$. At high SNR, the peak energy $z_{r_q}$ would also exceed the threshold $\mu$ with high probability (see (\ref{equation39})). In other words, the single-UT term $\mbox{Pr}(Q = 1 \, \vert \, Q \geq 1 )  \mbox{Pr}(\mbox{Timing error} \, | \, Q = 1)$ in the R.H.S of (\ref{equation323}) can be made as small as possible by choosing a sufficiently high SNR $\rho$.      
Next, for a desired false alarm probability we derive the expression for the threshold $\mu$.                
\begin{theorem}
\label{thmu}
In the single-UT scenario with a rectangular receiver window $W_{rx}[n] = 1, n=0,1,\cdots, (N-1)$, the threshold $\mu$ required to
achieve a desired false alarm probability $p_{fa}$ is

{\vspace{-5mm}
\small
\begin{eqnarray}
\label{mupfaeqn}
\mu(p_{fa}) & = & -N_o (1 + \lceil G B_c \rceil) \left\lfloor \frac{B_cT_c}{1 + \lceil G B_c \rceil} \right\rfloor    \log \left[  1 - \left(  1 - p_{fa}\right)^{\frac{1}{N_1 (1 + \lceil G B_c \rceil)} }  \right].
\end{eqnarray}
\normalsize}
\end{theorem}
\begin{proof}
See Appendix \ref{appendixC}.
\end{proof}
From Theorem \ref{thmu} it is clear that for a given $p_{fa}$, the required threshold $\mu$ decreases with decreasing $N_o$ (i.e., increasing SNR $\rho$). This substantiates our
claim (in the paragraph before Theorem \ref{thmu}) that at high $\rho$ the peak received RA preamble energy $z_q$ would exceed $\mu$ with high probability. 
\subsection{Multi-UT scenario}
\label{secmultiUT}
In a multi-UT scenario, as each UT {\em randomly} chooses one of the $R$ RA preambles and transmits it, 
two or more UTs could transmit the same RA preamble (i.e., a collision event). In a collision, it is clear that even at high SNR $\rho$, there will be timing error for a UT whose
received RA preamble power is smaller than that of some other UT which has transmitted the same RA preamble.
Based on this, the following Theorem gives a lower bound on the timing error probability (TEP). 
\begin{theorem}
\label{thm44}
If all UTs have i.i.d. fading and location statistics, then the average TEP for a given UT is lower bounded by

{\vspace{-5mm}
\small
\begin{eqnarray}
\label{eqnPebnd}
P_e & > & \sum\limits_{k=2}^\infty \mbox{Pr}(Q = k \, | \, Q \geq 1)   \left[  \sum\limits_{p=2}^k \, C(k-1, p-1) \, \frac{(p-1)}{p} \frac{(R - 1)^{k - p}}{ R^{k-1}} \right] \nonumber \\
& = &  \sum\limits_{k=2}^\infty \mbox{Pr}(Q = k \, | \, Q \geq 1) \, \left[   1 - \frac{R^k - (R-1)^k}{k  \, R^{k-1}} \right]  \,\,,\,\,  \nonumber \\
& &  C(m,n)  \Define  \frac{m!}{n! (m -n)!}
\end{eqnarray}
\normalsize}
where $R$ is the number of allowed RA preambles and the random variable $Q$ is the number of RA requests received in an OTFS frame.
\end{theorem}
\begin{proof}
See Appendix \ref{appendix_E}.
\end{proof}
\begin{corollary}
\label{corollary1}
For given probabilities $\{ \mbox{Pr}(Q = k \, | \, Q \geq 1) \}_{k=1}^{\infty}$, the lower bound in (\ref{eqnPebnd}) decreases monotonically with increasing number of RA preambles $R$.  
\end{corollary}
\begin{proof}
See Appendix \ref{prfcorollary1}.
\end{proof}
This result is expected since the chance of collisions decreases with increasing $R$. 
We note that, the lower bound in (\ref{eqnPebnd}) is independent of the SNR $\rho$ and depends only on $R$. Therefore, at high SNR, a timing error event can happen primarily due to two reasons, i) ``collisions'', and ii) ``interference'' from RA preamble transmission on an adjacent DDRE group in the DD domain. To reduce the possibility of strong interference from an adjacent RA preamble, the width $N_1$ of each RA preamble DDRE group needs to be sufficiently large.
Let the $q'$-th UT have a single channel path with Doppler shift $\nu > 0$. If this UT transmits the $q$-th RA preamble, then from the expression of $U_{q',1}[k] \,,\, k=0,1,\cdots, (N-1) $ in (\ref{xkleqndef}) it is clear that with the rectangular window, $\vert U_{q',1}[k] \vert$ is given by

{\vspace{-2mm}
\small
\begin{eqnarray}
\vert U_{q',1}[k] \vert = \left\vert  \frac{Sinc  \left[ N \left( \frac{k_{q} - k}{N}  + \frac{\nu}{\Delta f}\right) \right]}{Sinc \left(  \frac{k_{q} - k}{N}  + \frac{\nu}{\Delta f}  \right)}  \right\vert
\end{eqnarray}
\normalsize}
which has maximum magnitude at either $k = k_{q} +  \lfloor \frac{\nu}{\Delta f/N} \rfloor $ or $k = k_{q} + \lceil \frac{\nu}{\Delta f/N} \rceil $ and therefore $k - k_{q} \approx \nu/(\Delta f/N) = \nu N T$. Hence it is clear that most of the RA preamble energy along the Doppler domain will be received at DDRE locations roughly $ \lfloor \nu N T \rfloor$ or $ \lceil \nu N T \rceil$ DDREs away from the transmit DDRE location.  
Since the Doppler shift $\nu$ is in general not an integer multiple of $\Delta f/N$, energy will also be received on the $k = k_q + (1 + \lceil \nu NT \rceil)$-th DDRE. As the Doppler shift can be both positive and negative and $\vert \nu \vert \leq \nu_{max}$, the $q$-th RA preamble group ${\mathcal R}_q$ must contain DDREs having Doppler domain indices $\{ k_q - (1 + \lceil \nu_{max} NT \rceil) , k_q -  \lceil \nu_{max} NT \rceil, \cdots, k_q -1, k_q, k_q + 1, \cdots, k_q + (1 + \lceil \nu_{max} NT \rceil) \}$ (i.e., $2\left\lceil \nu_{max} N T \right\rceil  + 3$ DDREs) in order that the interference to adjacent DDRE groups is small.
We therefore propose the width $N_1$ of each DDRE group along the
Doppler domain to satisfy
\begin{eqnarray}
\label{N1cnstr}
2\left\lceil \nu_{max} N T \right\rceil  + 3 & \leq \,\, N_1 & \hspace{-2mm} \leq  N.
\end{eqnarray}
We illustrate the choice of $N_1$ through the following example. Let us consider a cell having radius $1500$ m and where UTs cannot lie within a $100$ m distance from the BS. Let $B_c = 1.08$ MHz, $T_c = 1.6$ ms, $G = 15 \mu s$ and $\nu_{max} = 300$ Hz. Therefore from (\ref{mtvaleqn}) and (\ref{Neqn1}) we have $M=18, N=96$, $\Delta f = 60 $ KHz and $T = 1/\Delta f = 16.66 \mu s$. Let the $q$-th RA preamble be transmitted by the $q'$-th UT situated $1000$ m from the BS. Let $(k_q,l_q) = (71,1)$. The channel from this UT to the BS consists of only a single-path for which $\tau_{q',1} = 2 \times 1000/(3 \times 10^8) = 6.66 \mu s$. This corresponds to a delay domain shift by approximately $\tau_{q',1}/(T/M) = 7.2$ DDREs. 
Hence, along the delay domain most of the received energy is localized around the $\lfloor l_q + \frac{\tau_{q',1}}{T/M} \rceil= 8$-th DDRE ($\lfloor x \rceil$ rounds $x$ to the nearest integer).
From (\ref{xkleqn}) it follows that the Doppler domain distribution of the $q$-th RA preamble energy received from the $q'$-th UT along the $i$-th channel path depends on $U_{q,i}[k] =  \frac{1}{N}  \sum\limits_{n=0}^{N-1} \left( W_{rx}[n] e^{j 2 \pi n \left( \frac{k_{q} }{N} + \frac{\nu_{q',i}}{\Delta f}  \right)} \right) e^{-j2 \pi \frac{n k}{N}}, k=0,1,\cdots, N-1$ (see (\ref{xkleqndef})), which is the $N$-point Discrete Fourier Transform (DFT) of the discrete time-domain sequence $e^{j 2 \pi n \left( \frac{k_{q} }{N} + \frac{\nu_{q',i}}{\Delta f}  \right)}, n=0,1,\cdots, N-1$, windowed with the time-domain waveform $W_{rx}[n], n=0,1,\cdots, N-1$. In traditional spectral analysis of time-domain signals, windowing waveforms have been used to achieve a trade-off between spectral resolution (main-lobe width) and the energy outside the main-lobe \cite{FHarris}. Therefore, for our system, the window $W_{rx}[\cdot]$ can be used to achieve a trade-off between the number of DDREs over which most of the received $q$-th RA preamble energy is spread in the Doppler domain (i.e., main-lobe width) and the maximum side-lobe energy level in the Doppler domain. The main-lobe width determines the width of each RA preamble along the Doppler domain (i.e., $N_1$) and the maximum side-lobe energy level determines the interference caused to other RA preambles.
For better understanding of this trade-off, in Fig.~\ref{fig17} we plot the normalized $\vert {\widehat x}[k,l=8] \vert^2$ along the Doppler domain i.e., versus $k$ for different receiver windows $W_{rx}[\cdot]$, namely the rectangular window, the Hamming window and two different types of Blackman-Harris window \cite{FHarris}.\footnote{\footnotesize{The rectangular window is given by $W_{rx}[n] = 1, n=0,1,\cdots, (N-1)$ and the Hamming window is given by $W_{rx}[n] = \alpha_N (0.54 - 0.46 \cos(2 \pi n/N)) \,,\, n=0,1,\cdots, (N-1)$. The three-term Blackman-Harris window is given by $W_{rx}[n] =  0.42323 - 0.49755 \cos(2 \pi n/N) +  0.07922 \cos(4 \pi n/N)$ and the four-term Blackman-Harris window is given by $W_{rx}[n] = 0.35875 - 0.48829 \cos(2 \pi n/N) +  0.14128 \cos(4 \pi n/N) -  0.01168 \cos(6 \pi n/N) \,,\, n=0,1,\cdots,(N-1)$. For each window, $\alpha_N > 0$ is a constant such that $\sum\limits_{n=0}^{N-1} W_{rx}^2[n] \, = \, N$\cite{FHarris}.}}

In order to understand interference between adjacent RA preambles, we consider no AWGN in Fig.~\ref{fig17}. The Doppler shift of the channel path is taken to be $\nu_{q,1} =  \nu_{max} = 300$ Hz. Since $k_q + \frac{\nu_{q,1}}{\Delta f/N} = 71.48$, most of the received energy is localized to $k=71$ and $k=72$ for all the windows (see Fig.~\ref{fig17}). There is however leakage of the RA preamble energy to the $k=73$-rd DDRE (the leakage energy at $k=73$ is
about $17$ dB below the peak energy received at $k=71,72$ for the Hamming window). For the Hamming window, the leakage energy at
DDRE locations beyond the $k=73$-rd DDRE is more than $40$ dB below the peak energy received at $k=71,72$. In a near-far scenario, a UT near the BS (i.e., at a distance of $100$ m from the BS) transmits a RA preamble on a DDRE group which is adjacent (along the Doppler domain) to the DDRE group of another RA preamble transmitted by a UT at the cell-edge. With a path-loss exponent of $3$, the difference in the received RA preamble energy levels from the near and the far UT is $10 \log_{10} (1500/100)^3 = 35.3$ dB. Since the transmitted energy level is the same for all UTs, the width $N_1$ of a RA preamble must be such that the leakage from a RA preamble transmitted by a UT near the BS to an adjacent RA preamble transmitted by a cell-edge UT is at least  $35.3$ dB less than the received RA preamble energy of the UT near the BS. Therefore, with the Hamming window, it suffices to have only the Doppler domain locations $k=71,72, 73$ to be part of the DDRE group ${\mathcal R}_q$ corresponding to the transmitted RA preamble. Since, the Doppler shift can also be $\nu_{q',1} = - \nu_{max}$, the DDRE locations $k=69, 70$ must also be a part of the DDRE group. Hence, ${\mathcal R}_q = \{ (k,l) \,| \, l=0,1,\cdots, (M-1) \,,\, 69 \leq k \leq 73 \}$ and therefore $N_1 = 5$. This value of $N_1 = 5$ is equal to the
proposed lower bound to $N_1$ in (\ref{N1cnstr}), since $2\left\lceil \nu_{max} N T \right\rceil  + 3 = 2 \left\lceil 0.48 \right\rceil + 3 = 5$.
\begin{figure}[h]
	\vspace{-0.2 cm}
	\centering
	\includegraphics[width= 5.0 in, height= 3.5 in]{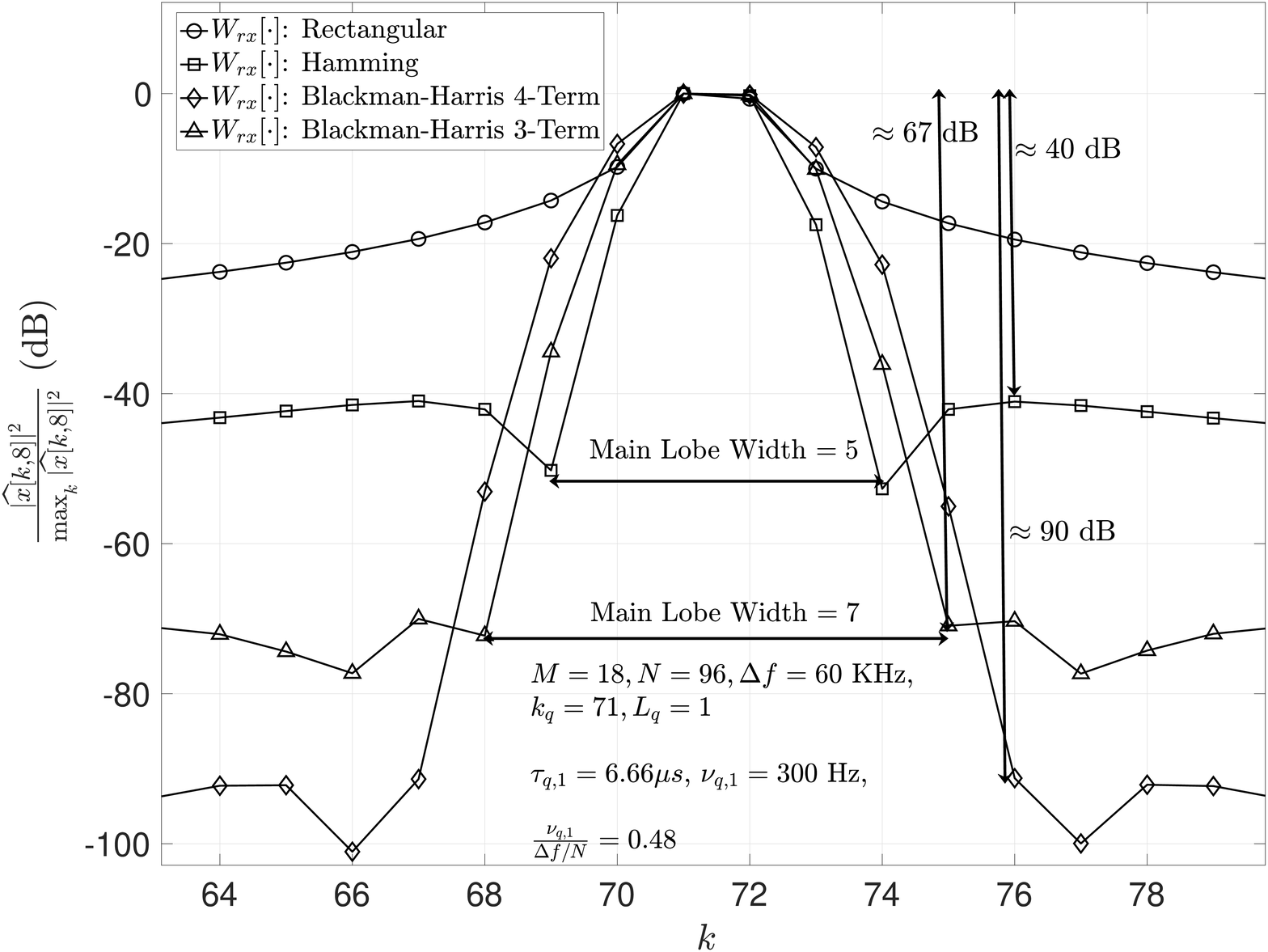}
	\vspace{-0.2 cm}
	\caption{Normalized $\vert {\widehat x}[k,l=8] \vert^2$ versus $k$.} 
	\vspace{-0.2cm}
	\label{fig17}
\end{figure}
With the Hamming window the leakage energy is still only about $40$ dB below the received RA preamble energy and therefore in a larger cell (radius $> 1500$ m), the near-far scenario could lead to a situation where the energy leaked from the received RA preamble from a UT near the BS is more than the RA preamble energy received from a cell-edge UT.
For larger cells we can use the $3$-term Blackman Harris window since the leakage energy is more than $67$ dB below the received RA preamble energy and therefore with a $5$ dB margin for fading, the cell radius must be such that the difference in the received energy levels of RA preambles transmitted from a cell edge UT and that from a UT near the BS is less than $67 - 5 = 62$ dB, i.e., $10 \log_{10}[(r_c/100)^3]  < 62$ where $r_c$ is the cell radius. Equivalently, $r_c < 100 \times 10^{6.2/3} =  11.6$ Km. However, from Fig.~\ref{fig17} it is also observed that with the $3$-term Blackman-Harris window the width of each RA preamble has to be increased to $N_1 = 7$ when compared to $N_1 = 5$ with the Hamming window. This would then reduce the number of allowed RA preambles $R = \lfloor N/N_1 \rfloor$. Decrease in $R$ would in-turn lead to increased probability of collisions and therefore a higher TEP. This observation regarding the increase in the RA preamble width $N_1$ due to more stringent leakage requirement is a manifestation of the fundamental trade-off between the highest side-lobe energy level and the equivalent noise-bandwidth of windowing waveforms used for spectral analysis \cite{FHarris}.                                 
Hence, for the multi-UT scenario there is a trade-off between choosing $N_1$ large enough so as to reduce interference from adjacent DDRE groups but at the same time it should not be so large that the probability of collisions becomes the dominant component.        
\section{Selection of Design Parameters for the Proposed OTFS Based RA Waveform}
\label{sectionSelParam}
Let $\mu_Q$ be the average number of RA requests received in an OTFS frame, $P_e$ denotes the desired timing error probability (TEP) and let $P_{fa}$ be the desired false alarm probability. In this section, for a given set of system parameters $(\nu_{max}, G, B_c, T_c, \mu_Q, P_e, P_{fa})$ we propose the selection of the design parameters $(T, \Delta f, M, N, N_1, \rho, \mu)$. Firstly, in step (1), $(T, M ,N)$ are chosen based on (\ref{mtvaleqn}) and (\ref{Neqn1}), then $\Delta f = 1/T$.
From Section \ref{secsingleUT} we know that the first term in the TEP expression in the R.H.S of (\ref{equation323}) corresponds to the TEP in the single-UT scenario and depends primarily on the SNR $\rho$. Similarly, from Section \ref{secmultiUT} we know that at high SNR, $N_1$ decides the contribution to the TEP from the remaining terms $k=2,3,\cdots, \infty$ in the R.H.S of (\ref{equation323}).    
From Section \ref{secmultiUT} we know that the choice of $N_1$ should reduce interference between adjacent RA preambles and also ensure low probability of collisions. Therefore in step (2), assuming no AWGN (i.e., very high $\rho$), through numerical simulations we find the smallest possible $N_1$ which achieves a TEP sufficiently smaller than
the desired $P_e$ (say half of the desired $P_e$).
Next, in step (3), through numerical simulations we select the operating $\rho$ such that the overall TEP (i.e., sum of all terms in the R.H.S. of (\ref{equation323})) is less than or equal to the desired TEP $P_e$ (for each value of $\rho$, the false alarm probability threshold $\mu$ is chosen so as to guarantee the desired false alarm probability $P_{fa}$).

This design procedure will be successful in achieving the desired $P_e$, only if $P_e$ is larger than the lower bound (due to collisions) given by (\ref{eqnPebnd}) in Theorem \ref{thm44}.
From Section \ref{secsingleUT} and Section \ref{secmultiUT} it is clear that with an appropriately chosen $N_1$, the interference due to an adjacent RA preamble can be reduced effectively and hence at high SNR, a timing error event can happen primarily due to collisions. It therefore appears that with the proper choice of $N_1$ the lower bound to the TEP as proposed in Theorem \ref{thm44} is in fact the limiting value of TEP as $\rho \rightarrow \infty$. This has in fact been verified through exhaustive simulations, some of which have been presented in Section \ref{simsection}. This lower limit however depends only on the number of RA preambles $R$ and the probabilities $\mbox{Pr}(Q = k \, | \, Q \geq 1)$. 
These probabilities depend on the spatial density of the RA call rate which we denote by $\lambda$ (in calls/sec/sq. Km), the cell-size and the time duration $T_a$ allowed between
two successive RA requests from a UT (e.g., $10$ ms in 4G-LTE). These probabilities also depend on the statistics of the arrival process. In this paper, we model the random number of simultaneous RA requests (i.e., $Q$) to be Poisson distributed. For a circular cell with prohibited inner radius $r_a$ and outer radius $r_c$, the mean value of $Q$ is given by

{\vspace{-5mm}
\small
\begin{eqnarray}
\label{muQeqn}
\mu_Q \Define {\mathbb E}[Q] & = & \pi (r_c^2 - r_a^2) \lambda T_a.
\end{eqnarray}
\normalsize}
As $Q$ is Poisson distributed we have $\mbox{Pr}(Q = k) = \frac{e^{-\mu_Q} \mu_Q^k}{k !}$ and therefore

{\vspace{-5mm}
\small
\begin{eqnarray}
\mbox{Pr}(Q = k \, | \, Q \geq 1)   =    \frac{\mbox{Pr}(Q = k)}{\mbox{Pr}(Q \geq 1)}  =    \frac{e^{-\mu_Q} \mu_Q^k}{(1 - e^{-\mu_Q}) k!}    \,,\, k \geq 1.
\end{eqnarray}
\normalsize}
\begin{figure}[h]
	\vspace{-0.2 cm}
	\centering
	\includegraphics[width= 5.0 in, height= 3.5 in]{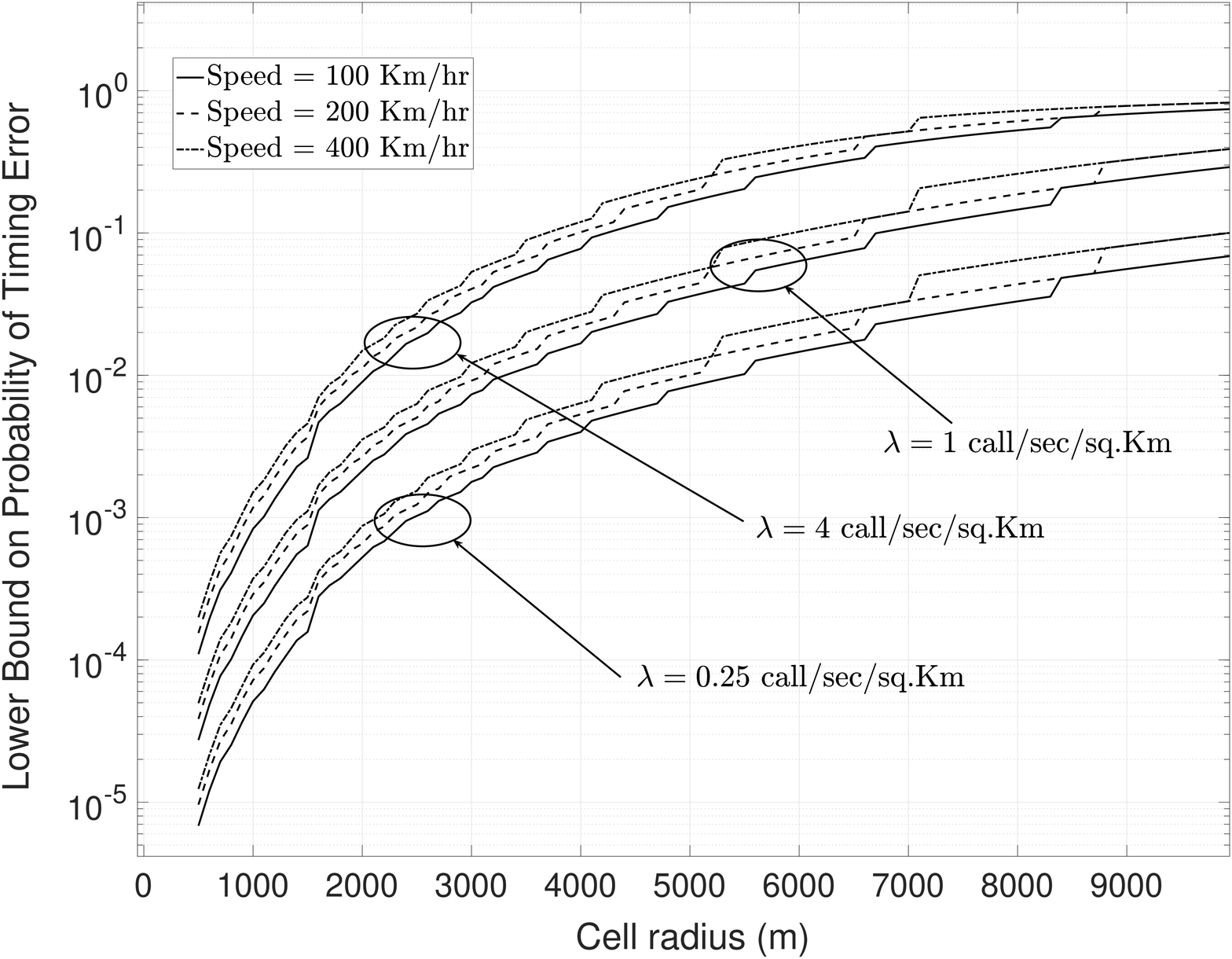}
	\vspace{-0.2 cm}
	\caption{Lower bound to the timing error probability (TEP) (R.H.S. of (\ref{eqnPebnd})) versus cell radius.} 
	\vspace{-0.2cm}
	\label{fig21}
\end{figure}
Next, in Fig.~\ref{fig21} we plot the lower bound to the TEP given by the R.H.S. in (\ref{eqnPebnd}) versus increasing cell-radius $r_c$ ($500 \leq r_c  \leq 10000$ m) for different
combinations of mobile speed ($100, 200, 400$ Km/hr) and call rate spatial density $\lambda = 0.25, 1.0, 4.0$ calls/sec/sq. km. We consider $B_c = 1.08$ MHz, $T_c = 1.6$ ms. The carrier frequency is taken to be $f_c = 4 \times 10^9$ Hz and therefore the maximum Doppler shift is $\nu_{max} = v f_c /(3 \times 10^8)$ where $v$ is the mobile speed in m/s. The maximum round-trip propagation delay is taken to be $G = \frac{2 \, r_c} {3 \times 10^8}$. The other parameters $(M,N,T, \Delta f)$ are chosen based on (\ref{mtvaleqn}) and (\ref{Neqn1}). As observed in Section \ref{secmultiUT}, for a path-loss exponent of $3$ and inner radius $r_a = 100$ m, we had proposed the use of the Hamming window with $N_1= 2\left\lceil \nu_{max} N T \right\rceil  + 3$ when $r_c \leq 1500$ and the $3$-term Blackman-Harris window with $N_1 = 2\left\lceil \nu_{max} N T \right\rceil  + 5$ when $1500 < r_c   \leq 10000$. We consider $T_a = 10^{-2}$ second.  

From Fig.~\ref{fig21} it is clear that, for a given mobile speed and a given $\lambda$, TEP increases with increasing cell radius. This is expected, since
with increasing cell radius the leakage from an adjacent RA preamble in a near-far scenario can be more severe (see the discussion in Section \ref{secmultiUT}).    
Additionally, with increase in cell radius, the average number of simultaneous RA calls in one OTFS frame i.e., $\mu_Q = \pi (r_c^2 - r_a^2) \lambda T_a$ increases and which therefore increases the TEP. However, the most interesting observation is that, for a given cell radius and a given $\lambda$, the TEP increases only slightly even when the mobile speed increases from $100$ Km/hr to $400$ Km/hr. This demonstrates the robustness of the proposed OTFS based RA preamble waveforms against Doppler spread. 
\section{Numerical Results}
\label{simsection}
We consider a single-cell scenario with a cell radius of $1500$ m and a single-antenna BS at the
centre of the cell. Single-antenna UTs are uniformly distributed in the cell, except inside a circular
region of radius $100$ m around the BS. The number of UTs transmitting RA preambles simultaneously
in the same OTFS frame is Poisson distributed. For the channel from a UT
to the BS, we consider the tapped delay line model in 3GPP Extended Typical Urban (ETU) channel
\cite{3GPPmdl} where the delay profile $\{\tau_{q,i}\}_{i=1}^{L_q=9}$
 is $ \tau_q + [0, 50, 120, 200, 230, 500, 1600, 2300, 5000]$
ns (i.e., a delay spread of $5 \mu s$) and the corresponding power profile $\left\{ {\mathbb{E}}\left [\left|h_{q,i}\right|^2 / \beta_q \right] \right\}_{i=1}^{L_q=9}$ is $[-1, -1, -1, 0, 0, 0, -3, -5, -7]$
dB. Here $\tau_q$ ns is the smallest path delay between the $q$-th UT and the BS. The channel path gains $h_{q,i}$ in (\ref{equation1}) are modeled as independent Rayleigh faded random variables and the power profile is normalized to satisfy (\ref{equation2}). The path-loss exponent is taken to be three and the reference path-loss for a cell-edge UT is unity, i.e., at the BS the received RA preamble power to noise power ratio is $\rho$ for a RA preamble transmitted by a cell-edge UT whereas it is $\rho \times (1500/100)^3$ for a UT at a distance of $100$ m from the BS.
For the $q$-th UT, the Doppler shift for
the $i$-th propagation path is taken to be $\nu_{q,i} = \nu_{max} \cos{(\theta_{q,i})}$ where $\theta_{q, i}$ are independent and
uniformly distributed in $[0,  2\pi)$ \cite{emanuele} and $\nu_{max}$ is the maximum possible Doppler shift.
For the purpose of comparison we also consider a 4G-LTE type system where $B_c =1.08$ MHz and the RA preamble is as per LTE preamble format $2$ which consists of a $1.6$ ms long Zadoff-Chu (ZC)
sequence \cite{4GLTE}. 
For the proposed OTFS based RA preamble transmission also we consider $T_c= 1.6$ ms and $B_c = 1.08$ MHz.

\begin{figure}[h]
	\vspace{-0.2 cm}
	\hspace{-0.2 in}
	\centering
	\includegraphics[width= 5.0 in, height= 3.5 in]{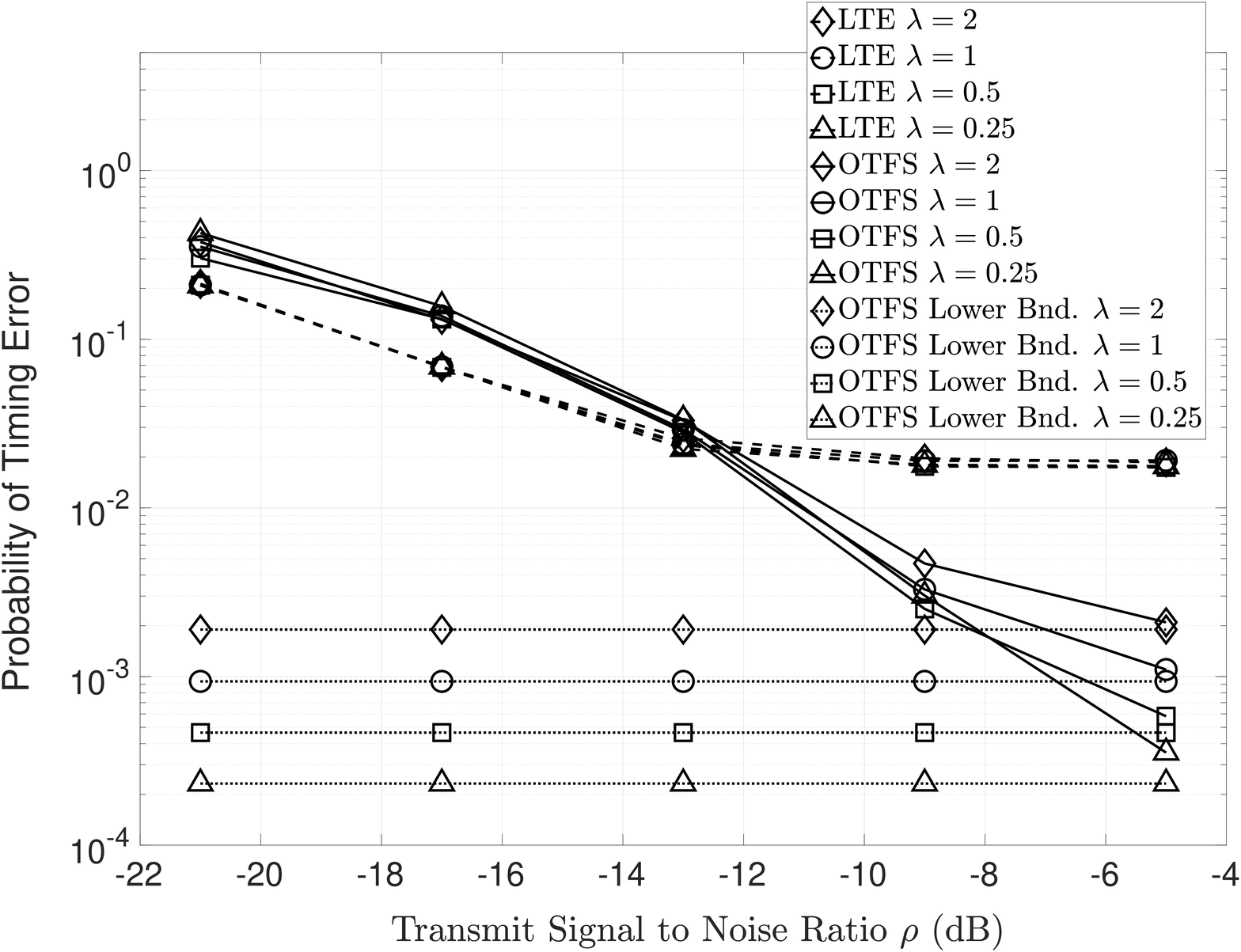}
	\vspace{-0.1 cm}
	\caption{Timing Error Probability (TEP) of the proposed OTFS based RA and that of 4G-LTE RA vs. SNR $\rho$ for $\lambda = 0.25, 0.5, 1.0, 2.0$ calls/sec/sq. Km. $B_c = 1.08$ MHz, $T_c = 1.6$ ms, $P_{FA} = 1 \times 10^{-2}$, $\nu_{max} = 300$ Hz.} 
	\vspace{-0.2 cm}
	\label{fig8}
\end{figure}

Next, in Fig.~\ref{fig8}, for a fixed desired false alarm probability $P_{FA} = 1 \times 10^{-2}$ and $\nu_{max} = 300$ Hz,  we compare the             
timing error probability (TEP) of the proposed OTFS based RA and that of a 4G-LTE based RA as a function of increasing SNR ($\rho$) for different RA call rate spatial densities $\lambda = 0.25, 0.5, 1.0, 2.0$ calls/sec/sq. Km. The TEP for the proposed OTFS based RA is given by (\ref{equation39}). For a given $\lambda$, the mean value of $Q$ i.e., $\mu_Q$ is given by (\ref{muQeqn}) where we consider $T_a = 1 \times 10^{-2}$ sec. The maximum round-trip propagation delay is taken to be $G = 2 \times 1500/(3 \times 10^8) + 5 \times 10^{-6}$ sec and therefore from (\ref{mtvaleqn}) and (\ref{Neqn1}) we have $M=18, N= 96, \Delta f = 60$ KHz and $T = 16.66 \mu s$. Based on the discussion in Section \ref{secmultiUT}, the width $N_1$ of each DDRE group is taken to be $2\left\lceil \nu_{max} N T \right\rceil  + 3 = 5$ (see (\ref{N1cnstr})) and the Hamming window is used at the receiver. The number of RA preambles is $R = \lfloor N/N_1 \rfloor = 19$. In Fig.~\ref{fig8}, for the OTFS based RA method, we also plot the proposed collision based lower bound to the TEP in (\ref{eqnPebnd}). 
From Fig.~\ref{fig8} it is clear that at sufficiently high $\rho$, for all considered values of $\lambda$, the TEP of the proposed OTFS based RA method is significantly smaller than that of the 4G-LTE RA method.
For the proposed OTFS based RA, as discussed in Section \ref{secmultiUT}, at high $\rho$ the TEP is dominated by collision events. Indeed, in Fig.~\ref{fig8}, at a high $\rho = -5$ dB, the TEP is close to the analytical lower bound given by (\ref{eqnPebnd}).

\begin{figure}[h]
	\vspace{-0.2 cm}
	\hspace{-0.2 in}
	\centering
	\includegraphics[width= 5.0 in, height= 3.5 in]{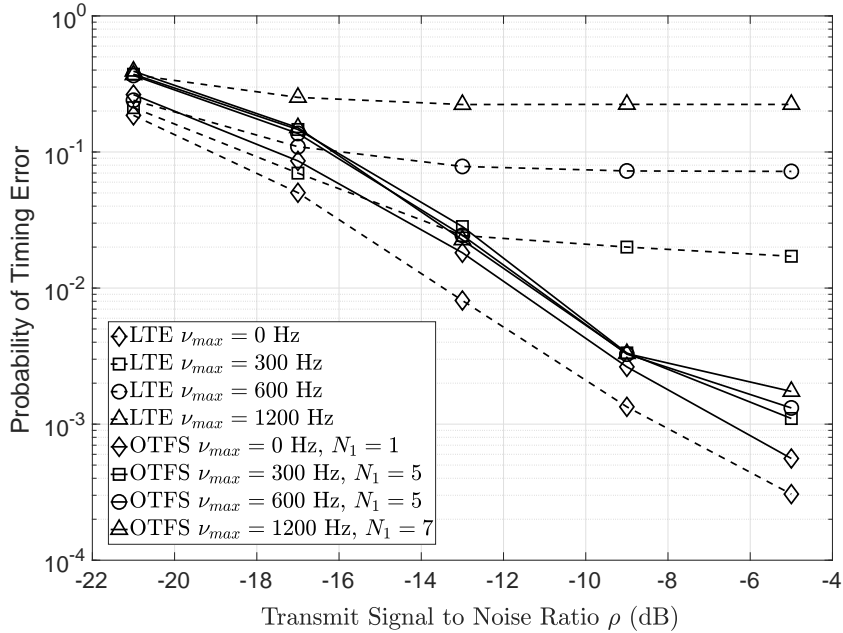}
	\vspace{-0.1 cm}
	\caption{Timing Error Probability (TEP) of the proposed OTFS based RA and that of 4G-LTE RA vs. SNR $\rho$ for RA call rate spatial density $\lambda = 1.0$ call/sec/sq. Km. $B_c = 1.08$ MHz, $T_c = 1.6$ ms, $P_{FA} = 1 \times 10^{-2}$, $\nu_{max} = 0,300,600,1200$ Hz.} 
	\vspace{-0.2 cm}
	\label{fig9}
\end{figure}

In Fig.~\ref{fig9} we plot the TEP of the proposed OTFS based RA and the RA method in 4G-LTE versus SNR $\rho$ for $\nu_{max} = 0, 300, 600, 1200$ Hz, a fixed $\lambda = 1$ call/sec/sq. Km, and
a fixed desired false alarm probability $P_{FA} = 1 \times 10^{-2}$.
For the proposed method, we choose $N_1 = 1$ when $\nu_{max} = 0$ Hz and for all other values of $\nu_{max}$ we choose $N_1 =  2\left\lceil \nu_{max} N T \right\rceil  + 3$ (see (\ref{N1cnstr})).
The Hamming window is used at the receiver. 
From Fig.~\ref{fig9} it is clear that, with increase in Doppler shift, the TEP of the 4G-LTE RA method degrades significantly. In comparison, the TEP of the proposed OTFS based RA remains almost the same even when the maximum Doppler shift $\nu_{max}$ is increased from $300$ Hz to $600$ Hz and then to $1200$ Hz. This is primarily because, channel induced Doppler shift affects only the location of the received RA symbol energy along the Doppler domain and does not significantly affect the location of the received RA symbol energy along the delay domain. This is in turn due to the separation of the terms $U_{q,i}[k]$ and $V_{q,i}[l]$ in the expression for the received delay-Doppler domain RA signal ${\widehat x}[k,l]$ in (\ref{xkleqn}) (see discussion in Section \ref{secsingleUT}).

\begin{table}[t]
\centering
\caption{Timing Error Probability (TEP) vs. $(N_1, \nu_{max})$, $\lambda = 1$ call/sec/sq. Km, $\rho = -5$ dB.}
{\small
\begin{tabular}{|| c|c|c|c||}
\hline
& $N_1 = 1$  & $N_1 = 2$ &  $N_1 = 5$   \nonumber \\
\hline
OTFS Lower & $1.85 \times 10 ^{-4}$ & $3.7 \times 10^{-4} $ &  $9.4 \times 10^{-4} $   \nonumber \\
Bound in (\ref{eqnPebnd}) & & & \nonumber \\
\hline
$\nu_{max} = 0$ Hz  & $5.6 \times 10^{-4}$ & $1.0 \times 10^{-3} $ &  $1.0 \times 10^{-3} $  \nonumber \\
\hline
$\nu_{max} = 300$ Hz & $6.4 \times 10^{-4}$ & $1.2 \times 10^{-3} $ &  $1.1 \times 10^{-3} $  \nonumber \\
\hline
$\nu_{max} = 600$ Hz & $1.8 \times 10^{-3}$ & $1.3 \times 10^{-3} $ &  $1.3 \times 10^{-3} $  \nonumber \\
\hline
$\nu_{max} = 1200$ Hz & $3.3 \times 10^{-2}$ & $1.5 \times 10^{-3} $ &  $1.6 \times 10^{-3} $  \nonumber \\
\hline
\end{tabular}
\normalsize}
\label{tab1}
\end{table}

\begin{table}[t]
\centering
\caption{Timing Error Probability (TEP) of the proposed OTFS based RA vs. $(P_{FA}, \nu_{max})$, $\lambda = 1$ call/sec/sq. Km, $\rho = -5$ dB.}
{\small
\begin{center}
\begin{tabular}{|| c|c|c|c||}
\hline
& $P_{FA}=1.0 \times 10^{-1}$  & $P_{FA}=1.0 \times 10^{-2}$ &  $P_{FA}=1.0 \times 10^{-3}$   \nonumber \\
\hline
$\nu_{max} = 0$ Hz, $N_1 = 1$  &  $5.4 \times 10^{-4}$  & $5.6 \times 10^{-4}$  & $6.5 \times 10^{-4}$  \nonumber \\
\hline
$\nu_{max} = 300$ Hz, $N_1 = 5$  &  $9.7 \times 10^{-4}$  & $1.1 \times 10^{-3}$  & $1.3 \times 10^{-3}$  \nonumber \\
\hline
$\nu_{max} = 600$ Hz, $N_1 = 5$  &  $1.3 \times 10^{-3}$  & $1.3 \times 10^{-3}$  & $1.5 \times 10^{-3}$  \nonumber \\
\hline
$\nu_{max} = 1200$ Hz, $N_1 = 7$ &  $1.7 \times 10^{-3}$ & $1.7 \times 10^{-3}$  & $1.9 \times 10^{-3}$  \nonumber \\
\hline
\end{tabular}
\end{center}
\normalsize}
\label{tab2}
\end{table}

In Section \ref{secmultiUT} it was observed that a small $N_1$ would decrease the collision probability due to increase in the number of RA preambles, but at the same time this would also increase interference between adjacent DDRE groups leading to increase in TEP. To understand this tradeoff, in Table-\ref{tab1}, for a fixed $\lambda = 1$ call/sec/sq. Km, $\rho = -5$ dB and a fixed desired false alarm probability $P_{FA} = 1.0 \times 10^{-2}$, we list the TEP for DDRE group width $N_1 = 1,2,5$ and $\nu_{max} = 0, 300, 600, 1200$ Hz. For each $N_1$ we also list the proposed lower bound to the TEP in (\ref{eqnPebnd}), which is based on collision events. The Hamming window is used at the receiver. For a fixed $N_1$, the TEP increases with increasing $\nu_{max}$ as a higher $\nu_{max}$ implies more interference between adjacent DDRE groups. For small $\nu_{max}$, even $N_1= 1$ could be sufficient in terms of guaranteeing little interference between adjacent DDRE groups. In the current example, the width of each DDRE is $\Delta f/ N = 625$ Hz along the Doppler domain and for $\nu_{max} = 0, 300$ Hz it is observed that for a fixed $\nu_{max}$ the TEP increases as $N_1$ is increased from $N_1=1$ to $N_1 = 2$ as this increases the probability of collisions (i.e., $N_1 = 1$ is sufficient).
However, when $\nu_{max}$ is large (e.g. $\nu_{max} = 600,1200$ Hz in Table-\ref{tab1}), then $N_1 = 1$ may not be sufficient in limiting the interference between adjacent DDRE groups. For example, in Table-\ref{tab1}, for $\nu_{max} = 1200$ Hz the TEP is $3.3 \times 10^{-2}$ when $N_1 = 1$ but decreases significantly to $1.5 \times 10^{-3}$ when $N_1$ is increased to $2$. This is because, although the probability of collisions increases with increasing $N_1$ from $N_1 = 1$ to $N_1 = 2$, for a high $\nu_{max}$ it still results in reduction in the overall TEP due to reduction in interference between adjacent DDRE groups. However, with further increase in $N_1$, very little further reduction in interference is possible as $N_1$ is already ``sufficiently large'' and therefore the TEP may increase due to increase in the probability of collisions (as the 
number of RA preambles $R = \lfloor N/ N_1\rfloor$ will decrease with increase in $N_1$). From Table-\ref{tab1}, $N_1 = 2$ appears to be sufficient for $\nu_{max} = 600,1200$ Hz. However, in Section \ref{secmultiUT} we had proposed that $N_1 \geq 2 \lceil \nu_{max} N T \rceil + 3 = 7$ (for $\nu_{max} = 1200$ Hz). This discrepancy is due to the fact that in Section \ref{secmultiUT} we had considered the worst-case near-far scenario where one UT is near the BS and the other is at the cell-edge, whereas Table-\ref{tab1} shows the ``average'' TEP when UTs are uniformly distributed in the cell.

In Fig.~\ref{fig8} and Fig.~\ref{fig9} we have plotted the TEP for a fixed desired false alarm probability $P_{FA} = 1 \times 10^{-2}$.
Next, we study the variation in the TEP of the proposed OTFS based RA with decreasing false alarm probability $P_{FA}$. In Table-\ref{tab2}, we list the TEP of the proposed OTFS based RA for $P_{FA} = 1 \times 10^{-1}, 1 \times 10^{-2}, 1 \times 10^{-3}$ and $\nu_{max} = 0, 300, 600, 1200$ Hz. We consider a fixed RA call rate spatial density $\lambda = 1.0$ call/sec/sq. Km.
Further, we choose $N_1=1,5,5,7,$ respectively for $\nu_{max} = 0, 300, 600, 1200$ Hz.
From Fig.~\ref{fig8} we know that the TEP is close to its lower bound at $\rho = -5$ dB and therefore in Table-\ref{tab2} we consider a fixed $\rho = -5$ dB.
A Hamming window is used at the receiver.
From Table-\ref{tab2} it is observed that the TEP increases with decreasing $P_{FA}$.
However, the amount of increase in the TEP is small. In the following, we provide an explanation for this observation.

From the expression of the TEP in (\ref{equation39}) it is clear that a missed detection event for an RA preamble happens if either the maximum received RA preamble energy is less than the false alarm threshold $\mu$
or if the timing advance (TA) estimate does not lie between the smallest and the greatest round-trip propagation delays of the UT which transmitted that RA preamble.
In Fig.~\ref{fig8} we have seen that at high SNR ($\rho = -5$ dB), the TEP is close to the proposed collision based lower bound in (\ref{eqnPebnd}), i.e., at high SNR the TEP is dominated by collision events.

The proposed timing advance (TA) estimate is based on the delay domain index of the DDRE where maximum RA preamble energy is received,
and therefore in a collision scenario, RA preamble detection will most likely succeed for the UT whose transmitted RA preamble is received at the BS with highest energy.
The proposed TA estimate will then correspond to the round-trip propagation delay for this UT.
At high $\rho$, the received RA preamble energy from the other UTs (which also transmitted the same RA preamble) will also be high and will most likely exceed the false alarm threshold.
However, RA preamble transmission by most of these other UTs will still not get detected as the acquired TA estimate will not lie between the smallest and the greatest round-trip propagation delay of several of these UTs.\footnote{\footnotesize{This is because, the round-trip propagation delay
for the UT whose RA preamble is detected need not be the same as the round-trip propagation delay of other UTs which transmitted the same RA preamble.}}
In other words, at high SNR, for a given UT, a missed detection event happens mostly because the acquired TA estimate does not lie between the smallest and the greatest round-trip propagation delay for that UT, and not so much due to the false alarm threshold.

Therefore, at high SNR, the TEP is expected to be not so sensitive to change in the false alarm threshold. 
For a given SNR $\rho$, the required false alarm threshold will increase with decrease in the desired false alarm probability $P_{FA}$ (see (\ref{equation37})).
From the expression of the TEP in (\ref{equation39})  it is clear that this increase in the false alarm threshold will result in some increase in the TEP. 
However, as discussed above, at high SNR the TEP is dominated by collision events and is not so sensitive to an increase in false alarm threshold. This explains the observation in Table-\ref{tab2} that, the
TEP increases only slightly with decreasing $P_{FA}$. Further, for a given $P_{FA}$, the required false alarm threshold $\mu$ does not depend on collision events since
by definition a false alarm event happens for an RA preamble when no UT transmits that RA preamble but still the maximum energy received in the corresponding DDRE group is
greater than the false alarm threshold $\mu$.

\section{Conclusion}
In this paper we have considered the problem of uplink timing synchronization in OTFS modulation based wireless communication systems.
We have proposed a novel waveform design for the transmission of random access (RA) preambles based on OTFS modulation.
We have also proposed a method to estimate the round-trip propagation delay (also known as timing advance) between the BS and the UT
which transmitted the RA preamble. For a given probability of false alarm, we define the timing error probability (i.e., probability of missed detection of an RA preamble), as the
probability of the event that either the maximum received RA preamble energy is less than the false alarm threshold, or the proposed timing advance (TA) estimate for the RA preamble does not lie between the smallest and the greatest path delays of the UT which transmitted
that RA preamble. Our study reveals that in a multi-UT scenario, at high signal-to-noise ratio (SNR), the timing error probability (TEP) is limited by interference from adjacent RA preambles and
collision events (i.e., when two or more UTs transmit the same RA preamble). To reduce the impact of interference between adjacent RA preambles we have proposed the use of appropriate
time-domain windowing at the BS receiver. By analyzing the collision events, we have derived a lower bound expression for the TEP. Our analysis is supported by numerical simulations which confirm
that at high SNR, the TEP of the proposed OTFS based RA is significantly more robust to channel induced Doppler shift when compared to the TEP of 4G-LTE based RA.
Simulations also confirm that, at high SNR, the TEP of the proposed OTFS based RA is limited by collision events.       
In future we plan to extend the study done in this paper to multi-cell systems where the same RA preambles are used in neighbouring cells and where the BS
in each cell has multiple antennas.
\appendices

\section{Proof of Theorem \ref{thmu}}
\label{appendixC}
From (\ref{equation36}) and (\ref{equation37}), the false alarm probability for the single-UT scenario is given by 

{\vspace{-5mm}
\small
\begin{eqnarray}
P_{FA} & = & \mbox{Pr}\left( \max_{(k,l) \in {\mathcal R}_q }  \vert w[k,l] \vert^2 \, \geq \, \mu \right)
\end{eqnarray}
\normalsize
}
where ${\mathcal R}_q$ is given by (\ref{ddredef}).
With rectangular window $W_{rx}[n,m] = 1, n=0,1,\cdots, (N-1), m=0,1,\cdots, (M-1)$, from (\ref{Ynmeqn1}) and (\ref{xkleqn}) it follows that $w[k,l] \, \sim \, \mbox{i.i.d.}  \, {\mathcal C}{\mathcal N}(0, MNN_o)$.
Therefore $1 - P_{FA}  =   \Pi_{(k,l) \in {\mathcal R}_q }    \, \mbox{Pr}\left(  \vert w[k,l] \vert^2 \, < \, \mu \right),$ or equivalently

{\vspace{-5mm}
\small
\begin{eqnarray}
\label{pfamueqn}
P_{FA} & = &  1 \, - \, \Pi_{(k,l) \in {\mathcal R}_q }    \, \mbox{Pr}\left(  \vert w[k,l] \vert^2 \, < \, \mu \right) \nonumber \\
& = & 1  - \, \Pi_{(k,l) \in {\mathcal R}_q }    \left[ 1 - \mbox{Pr}\left(  \vert w[k,l] \vert^2 \, \geq  \, \mu \right)  \right]  \nonumber \\
&  \mya &   1  - \, \Pi_{(k,l) \in {\mathcal R}_q }    \left[ 1 - e^{- {\frac{\mu}{  \left( 1 + \lceil G B_c \rceil  \right)  \left\lfloor  \frac{B_c T_c }{1 + \lceil G B_c \rceil }  \right\rfloor N_0 } } } \right] \nonumber \\
&  \myb & 1 - \left[  1 - e^{- {\frac{\mu}{ \left( 1 + \lceil G B_c \rceil  \right)  \left\lfloor  \frac{B_c T_c }{1 + \lceil G B_c \rceil }  \right\rfloor  N_0 } } }  \right]^{N_1 \left( 1 + \lceil G B_c \rceil  \right) }
\end{eqnarray}
\normalsize}
where step (a) follows from the fact that $w[k,l] \, \sim \, \mbox{i.i.d.}  \, {\mathcal C}{\mathcal N}(0, MNN_o)$ and (\ref{mtvaleqn}), (\ref{Neqn1}). Step (b) follows from $\vert {\mathcal R}_q \vert = N_1 M = N_1 ( 1 + \lceil G B_c \rceil) $ (see (\ref{ddredef}) and (\ref{mtvaleqn})).
From (\ref{pfamueqn}) it is clear that $P_{FA}$ decreases monotonically with increasing $\mu$.
Hence for a desired $p_{fa}$, the unique threshold $\mu(p_{fa})$ which achieves this
false alarm probability, is given by (\ref{mupfaeqn}).

\vspace{-3mm}
\section{Proof of Theorem \ref{thm44}}
\label{appendix_E}
Since $\mbox{Pr}(\mbox{Timing error and Collision})$ is the probability of both the timing error and the collision event
happening simultaneously, a lower bound to the TEP $P_e$ is given by

{\vspace{-5mm}
\small
\begin{eqnarray}
\label{ineq1}
P_e & = & \mbox{Pr}(\mbox{{Timing error}})  >  \mbox{Pr}(\mbox{Timing error and Collision})  \nonumber \\
&  \hspace{-7mm}  = &  \hspace{-3mm}  \sum\limits_{k=2}^\infty  \, \mbox{Pr}(Q = k \, | \, Q \geq 1) \, \mbox{Pr}(\mbox{Timing error and Collision} \, | \, Q=k)   \nonumber \\
& \hspace{-7mm}  = & \hspace{-3mm}  \sum\limits_{k=2}^\infty   \mbox{Pr}(Q = k \, | \, Q \geq 1) {\Bigg \{}  \sum\limits_{p=2}^k \mbox{Pr}(\mbox{Timing error and} \, (p-1) \mbox{UTs}   \,\,  \mbox{Colliding with the given UT} \, | \, Q=k)  {\Bigg \}} \nonumber \\
&  \hspace{-7mm}   =   &   \hspace{-3mm}   \sum\limits_{k=2}^\infty  \, \mbox{Pr}(Q = k \, | \, Q \geq 1)  {\Bigg \{}   \sum\limits_{p=2}^k   {\Big [} \mbox{Pr}((p-1) \mbox{UTs Colliding with the}   \,\, \mbox{given UT} \, | \, Q=k) \, \times  \nonumber \\
& & \hspace{40mm} \mbox{Pr}( \mbox{Timing Error} \, | \, (p-1) \mbox{UTs Colliding}  \,\,  \mbox{with the given UT} \, \mbox{and} \,  Q=k) {\Big ]} {\Bigg \}}.
\end{eqnarray}
\normalsize}
When $(p-1)$ UTs collide with the given UT, the probability of no timing error for the given UT is $1/p$ since
the UTs are distributed uniformly in the cell and therefore the received RA waveform of the given UT can have the largest energy among all the colliding
$p$ waveforms with probability $1/p$. Therefore for $p \leq k$ we have

{\vspace{-5.5mm}
\small
\begin{eqnarray}
\label{ineq2}
\mbox{Pr}( \mbox{Timing Error} \, | \, (p-1) \mbox{UTs Colliding with} 
 \,\, \mbox{the given UT} \, \mbox{and} \,  Q=k)  =  1 -  (1/p).
\end{eqnarray}
\normalsize}
When $(Q=k)$, the probability that $(p-1)$ out of the $k$ UTs collide with the given UT is

{\vspace{-3.5mm}
\small
\begin{eqnarray}
\label{ineq3}
\mbox{Pr}((p-1) \mbox{UTs Colliding with the given UT} \, | \, Q=k) 
 =  C(k-1, p-1) \frac{(R - 1)^{k - p}}{ R^{k-1}}
\end{eqnarray} 
\normalsize}
where $C(m,n) \Define m!/(n! (m-n)!)$.
Using (\ref{ineq2}) and (\ref{ineq3}) in (\ref{ineq1}) we get

{\vspace{-4mm}
\small
\begin{eqnarray}
\label{eqn71}
P_e  &  >  &  \hspace{-3mm} \sum\limits_{k=2}^\infty \hspace{-1mm} {\Bigg \{} \mbox{Pr}(Q = k  |  Q \geq 1) \hspace{-1mm}  \left[  \sum\limits_{p=2}^k  \hspace{-1mm} C(k-1, p-1)  \frac{(p-1)}{p} \frac{(R - 1)^{k - p}}{ R^{k-1}} \right] {\Bigg \}}.
\end{eqnarray}
\normalsize}
Next, substituting $p$ by $l = p-1$ for the summation index in the R.H.S. of (\ref{eqn71}) we get

{\vspace{-4mm}
\small
\begin{eqnarray}
\label{prfeqnbnd}
\sum\limits_{p=2}^k \, C(k-1, p-1) \, \frac{(p-1)}{p} \frac{(R - 1)^{k - p}}{ R^{k-1}} & = &   \sum\limits_{l=1}^{k-1} \, C(k-1, l) \, \frac{l}{l+1} \frac{(R - 1)^{(k-1) - l}}{ R^{k-1}} \nonumber \\
&  \hspace{-115mm}  = & \hspace{-60mm}  \left( \frac{R-1}{R}\right)^{k-1} {\Bigg [}   \underbrace{\sum\limits_{l=0}^{k-1} \,   \frac{C(k-1, l)}{(R - 1)^l}}_{= \left( 1 + \frac{1}{R-1} \right)^{k-1}}   \, - \,   \sum\limits_{l=0}^{k-1} \, \frac{C(k-1, l)}{l+1}  \left( \frac{1}{R - 1}\right)^l   {\Bigg ]}  \nonumber \\
& \hspace{-115mm}  \mya &  \hspace{-60mm}   \left(\frac{R-1}{R}\right)^{k-1} \hspace{-1mm} {\Big [}    \left(\frac{R}{R-1}\right)^{k-1} \hspace{-4mm}  -   (R-1) \sum\limits_{l=0}^{k-1} \hspace{-1mm} C(k-1, l) \hspace{-1mm} \int_{0}^{\frac{1}{R-1}} \hspace{-1mm} x^l \, dx\,  {\Big ]} \nonumber \\
& \hspace{-115mm}  = &  \hspace{-60mm}   \left(\frac{R-1}{R}\right)^{k-1} \hspace{-1mm} {\Big [}    \left(\frac{R}{R-1}\right)^{k-1} \hspace{-4mm}  -   (R-1) \hspace{-1mm} \int_{0}^{\frac{1}{R-1}}  \hspace{-1mm} \underbrace{\sum\limits_{l=0}^{k-1} \hspace{-1mm} C(k-1, l)   x^l  }_{=  \left( 1 + x \right)^{k-1}}   dx\,  {\Big ]} \nonumber \\
&  \hspace{-115mm}  \myb & \hspace{-60mm}   1 \, - \, \frac{(R-1)^k}{R^{k-1}} \int_{0}^{\frac{1}{R-1}} \, (1 + x)^{k-1} \, dx = 1 - \frac{R^k - (R-1)^k}{k R^{k-1}}
\end{eqnarray}
\normalsize}
where step (a) follows from the binomial expansion of $\left( 1 + \frac{1}{R-1} \right)^{k-1}$ and also the fact that $\frac{1}{l+1} \frac{1}{(R- 1 )^l} = (R-1) \int_{0}^{\frac{1}{R-1}} \hspace{-1mm} x^l \, dx$.
Substitution of $\frac{1}{l+1} \frac{1}{(R- 1 )^l}$ by the integral expression $(R-1) \int_{0}^{\frac{1}{R-1}} \hspace{-1mm} x^l \, dx$, helps us to derive a closed form expression in step (b).
Using (\ref{prfeqnbnd}) in (\ref{eqn71}) then gives (\ref{eqnPebnd}) which completes the proof.

\section{Proof of Corollary \ref{corollary1}}
\label{prfcorollary1}
It suffices to show that the first derivative (w.r.t $R$) of each term inside the summation in the lower bound expression in (\ref{eqnPebnd}), is negative for all $R \geq 1$.
This is equivalent to showing that the first derivative of $B_k(R) \Define \left( R^k - (R-1)^k \right)/ R^{k-1}$ w.r.t. $R$ is positive for all $R \geq 1$ and all positive integer $k \geq 2$.
The first derivative of $B_k(R)$ w.r.t. $R$  is given by
\begin{eqnarray}
\label{deriv2}
\hspace{-4mm} \frac{d}{d R}  \left[ \frac{R^k - (R-1)^k}{R^{k-1}} \right] & \hspace{-3mm}  = &  \hspace{-3mm} \frac{R^k - \left[ k + R - 1\right](R - 1)^{k-1} }{R^k}.
\end{eqnarray}
The numerator of the R.H.S. in (\ref{deriv2}) is further given by

\vspace{-4mm}
\small
\begin{eqnarray}
R^k - \left[ k + R - 1\right](R - 1)^{k-1}  & \hspace{-3mm}  =  & \hspace{-3mm} \left(1 + (R-1)\right)^k - k (R-1)^{k-1}  - (R-1)^k \nonumber \\
& \hspace{-74mm}  \mya  & \hspace{-38mm}  \overbrace{1 + k (R-1)^{k-1} + (R-1)^k + \sum\limits_{p=1}^{k-2} C(k,p) (R - 1)^p}^{= \left(1 + (R-1)\right)^k}   \,  -  \, k (R-1)^{k-1}  - (R-1)^k \nonumber \\
& \hspace{-74mm}  =  & \hspace{-38mm}   1 + \sum\limits_{p=1}^{k-2} C(k,p) (R - 1)^p \, \mygtb \, 1
\end{eqnarray}
\normalsize
where in step (a) we have used the binomial expansion for $(1 + (R-1))^k$ and $C(k,p)$ is defined in (\ref{eqnPebnd}).
The inequality in step (b) follows from the fact that $R \geq 1$.
Using this inequality in (\ref{deriv2}) it follows that the first derivative of $B_k(R)$ is positive for all positive integer $k \geq 2$ and $R \geq 1$. This then completes the proof.

\newpage 

\begin{IEEEbiographynophoto}{Alok Kumar Sinha}(S'19)
received the B.Tech degree in Electronics and Telecommunication Engineering from BPUT, Odisha, India, in 2011 and the M.E. degree in Electronics and Communication from the Birla Institute of Technology, Mesra, India, in 2013. He is currently pursuing the Ph.D. degree from the Indian Institute of Technology, Delhi, India. His research interests include massive MIMO systems, wireless communication, and OTFS.
\end{IEEEbiographynophoto}
\vspace{-5mm}
\begin{IEEEbiographynophoto}{Saif Khan Mohammed}(SM'15) is currently an Associate Professor at the Department of Electrical Engineering, I.I.T. Delhi. He received the B.Tech degree in Computer Science and Engineering from the Indian Institute of Technology (I.I.T.), New Delhi, India, in 1998 and the Ph.D. degree in Electrical and Communication Engineering from the Indian Institute of Science, Bangalore, India, in 2010.  His main research interests include wireless communication using large antenna arrays, coding and signal processing for wireless communication systems.
He currently serves as an Editor for the IEEE Transactions on Wireless Communications and the IEEE Wireless Communications Letters.
Dr. Mohammed was awarded the 2017 NASI Scopus Young Scientist Award and the Teaching Excellence Award at I.I.T. Delhi for 2016-17. He is also a recipient of the Visvesvaraya Young Faculty Fellowship by the Ministry of Electronics and IT, Govt. of India (2016 - 2019).
\end{IEEEbiographynophoto}
\vspace{-5mm}
\begin{IEEEbiographynophoto}{Raviteja Patchava}
received the M.E. degree in telecommunications from the Indian Institute of Science, Bangalore, India in 2014, and the Ph.D. degree from Monash University, Australia in 2019. From 2014 to 2015, he worked at Qualcomm India Private Limited, Bangalore, on WLAN systems design. He is currently working as a postdoctoral research fellow with Monash University, Australia. His current research interests include new waveform designs for next generation wireless systems. He was a recipient of Prof. S. V. C. Aiya Medal from the Indian Institute of Science, Bangalore in 2014. He also received the best student paper award for his publication at the SPAWC 2018 conference in Greece.
\end{IEEEbiographynophoto}
\vspace{-5mm}
\begin{IEEEbiographynophoto}{Yi Hong}(S'00--M'05--SM'10) is currently a Senior lecturer at the Department of Electrical and Computer Systems Eng.,
Monash University, Melbourne, Australia.
She obtained her Ph.D. degree in Electrical Engineering and Telecommunications 
from the University of New South Wales (UNSW), Sydney, and received   
the {\em NICTA-ACoRN Earlier Career Researcher Award} at the {\em Australian Communication
Theory Workshop}, Adelaide, Australia, 2007. She currently serves on the Australian Research Council College of Experts (2018-2020). 
Dr. Hong was an Associate Editor for {\em IEEE Wireless Communication Letters} and {\em Transactions on Emerging Telecommunications Technologies (ETT)}.
She was the General Co-Chair of {\em IEEE Information Theory Workshop} 2014, Hobart; the Technical Program Committee Chair of
{\em Australian Communications Theory Workshop} 2011, Melbourne; and the Publicity Chair
at the {\em IEEE Information Theory Workshop} 2009, Sicily. She was a Technical Program Committee member for
many IEEE leading conferences. Her research interests include
communication theory, coding and information theory with applications to telecommunication engineering.
\end{IEEEbiographynophoto}
\vspace{-5mm}
\begin{IEEEbiographynophoto}{Emanuele Viterbo}(M'95--SM'04--F'11)
is currently a professor in the ECSE Department and an Associate Dean in Graduate Research at
Monash University, Melbourne, Australia.
From 1990 to 1992 he was with the European Patent Office, The Hague, The Netherlands, as a patent examiner
in the field of dynamic recording and error-control coding. Between 1995 and 1997 he held a post-doctoral
position in the Dipartimento di Elettronica of the Politecnico di Torino. In 1997-98 he was a post-doctoral research fellow
in the Information Sciences Research Center of AT T Research, Florham Park, NJ, USA.
From 1998-2005, he worked as Assistant Professor and then Associate Professor, in Dipartimento di Elettronica at Politecnico di Torino.
From 2006-2009, he worked in DEIS at University of Calabria, Italy, as a Full Professor.
Prof. Emanuele Viterbo is an ISI Highly Cited Researcher since 2009. He is Associate Editor of IEEE Transactions on Information Theory,
European Transactions on Telecommunications and Journal of Communications and Networks, and
Guest Editor for IEEE Journal of Selected Topics in Signal Processing: Special Issue Managing Complexity in Multiuser MIMO Systems.
Prof. Emanuele Viterbo was awarded a NATO Advanced Fellowship in 1997 from the Italian National Research Council.
His main research interests are in lattice codes for the Gaussian and fading channels, algebraic coding theory,
algebraic space-time coding, digital terrestrial television broadcasting, digital magnetic recording, and irregular sampling.
\end{IEEEbiographynophoto}

\end{document}